\newif\ifconf
\conftrue

\newif\ifcomm
\commtrue

\newif\ifblind
\blindfalse

\newif\ifacm
\acmfalse

\newif\ifacmart
\ifacm   
   \acmarttrue
\else
   \acmartfalse
\fi

\newif\ifs
\strue

\ifconf
    \newcommand{\Conf}[1]{#1}
    \newcommand{\TR}[1]{}
    \newcommand{\Journal}[1]{}  
    \newcommand{\OnlyTR}[1]{}   
\else
    \newcommand{\Conf}[1]{}
    \newcommand{\TR}[1]{#1}
    \newcommand{\Journal}[1]{}  
    \newcommand{\OnlyTR}[1]{#1}   
\fi

\ifacm
	\ifacmart
    	\documentclass[sigconf]{acmart}
    \else
      \Conf{
 	     \documentclass[letterpaper,conference]{sig-alternate-10pt}
      }
      \TR{
     	 \documentclass[letterpaper,conference]{sig-alternate}
      }
  \fi
\else
  \Conf{
  \documentclass[10pt,journal,twocolumn]{IEEEtran}
  }
  \TR{
      \documentclass[journal]{IEEEtran}
  }
  
\fi

\usepackage{graphicx}

\usepackage{amssymb}
\usepackage{amsfonts} 

\usepackage{bbm}
\usepackage{url}
\usepackage{algorithm}
\usepackage[noend]{algpseudocode}
\usepackage{times}
\usepackage{todonotes}
\usepackage{xspace}
\usepackage{balance}
\usepackage{enumitem}
\usepackage{booktabs}  
	\newcommand{\ra}[1]{\renewcommand{\arraystretch}{#1}}
\usepackage{array}
\newcolumntype{C}[1]{>{\centering\let\newline\\\arraybackslash\hspace{0pt}}m{#1}}
\ifacmart 
\else 
	\usepackage{cite} 
	\usepackage{lastpage} 
\fi

\usepackage{comment}

\usepackage{hyperref}
\definecolor{darkred}{rgb}{0.7,0,0}
\definecolor{darkgreen}{rgb}{0,0.5,0}
\hypersetup{colorlinks=true,
        linkcolor=darkred,
        citecolor=darkgreen}
 
\ifacm 

\else 
	\usepackage[cmex10]{amsmath}
	\usepackage{amsthm}  
    \newtheoremstyle{boldthm}{}{}{\itshape}{}{\bfseries}{.}{ }{\thmname{#1}\thmnumber{ #2}\thmnote{ (#3)}}
	\theoremstyle{boldthm}
\fi

\ifacmart
\else
  \newtheorem{theorem}{Theorem}
  \newtheorem{definition}{Definition}

\fi

\ifcomm
    \newcounter{commentNumberI}
     \newcommand{\IK}[1]{\addtocounter{commentNumberI}{1}{{({\color{blue} 							{(\arabic{commentNumberI}.)} Isaac: #1})}}} 
     \newcommand{\FIK}[1]{\footnote{\IK{#1}}}
     
     \newcommand{\Gal}[1]{\addtocounter{commentNumberI}{1}{{({\color{red} {(\arabic{commentNumberI}.)} Gal: #1})}}} 
     \newcommand{\GM}[1]{\Gal{#1}}
     \newcommand{\FGM}[1]{\footnote{\GM{#1}}}
     
      \newcommand{\Dean}[1]{\addtocounter{commentNumberI}{1}{{({\color{purple} {(\arabic{commentNumberI}.)} Dean: #1})}}} 
      \newcommand{\DL}[1]{\Dean{#1}}
      
      \newcommand{\DLA}[1]{\Dean{add:#1}}
      \newcommand{\dl}[1]{\footnote{\DL{#1}}}
      
       \newcommand{\Ariel}[1]{\addtocounter{commentNumberI}{1}{{({\color{green} {(\arabic{commentNumberI}.)} Ariel: #1})}}} 
       \newcommand{\FA}[1]{\footnote{\Ariel{#1}}}

	  \newcommand{\SV}[1]{\addtocounter{commentNumberI}{1}{{({\color{darkgreen} {(\arabic{commentNumberI}.)} Shay: #1})}}} 
	  \newcommand{\FSV}[1]{\footnote{\SV{#1}}}
      
       \newcommand{\Kathy}[1]{\addtocounter{commentNumberI}{1}{{({\color{green} {(\arabic{commentNumberI}.)} Kathy: #1})}}} 
       \newcommand{\KB}[1]{\Kathy{#1}}
       \newcommand{\FKB}[1]{\footnote{\KB{#1}}}

    \newcommand{\A}[1]{\textbf{[#1]}}       

    \newcommand{\F}[1]{\footnote {\color{red} [#1]}}

    \usepackage[normalem]{ulem}
    \newcommand{\del}[1]{\textcolor{red}{X: \textbf{\sout{#1}}}} 

\else
    \newcommand{\A}[1]{}
    \newcommand{\F}[1]{}
    \newcommand{\IK}[1]{}
    \newcommand{\FIK}[1]{}
    \newcommand{\SV}[1]{}
    \newcommand{\FSV}[1]{}
    \newcommand{\Gal}[1]{}
    \newcommand{\GM}[1]{}
	\newcommand{\FGM}[1]{}
    \newcommand{\Dean}[1]{}
    \newcommand{\DL}[1]{}
    
    \newcommand{\DLA}[1]{}
    \newcommand{\dl}[1]{}
    \newcommand{\Ariel}[1]{}
    \newcommand{\FA}[1]{}
    \newcommand{\KB}[1]{}
    \newcommand{\FKB}[1]{}

    \newcommand{\del}[1]{} 
    
\fi

\newcommand{\remove}[1]{}

\ifs
        \usepackage[footnotesize]{caption}
\else
	
    \usepackage{caption}
\fi

    \def\compactify{\itemsep=0pt \topsep=0pt \partopsep=0pt \parsep=0pt}
      \let\latexusecounter=\usecounter

\usepackage[labelformat=simple]{subcaption}
    
\newcommand{\sep}{\hspace*{.1em}}
\newcommand{\eq}[1]{(\ref{#1})}

\newcommand{\se}[1]{Sec.~\ref{#1}}
\newcommand{\fig}[1]{Fig.~\ref{#1}}
\newcommand{\al}[1]{Alg.~\ref{#1}}

\ifacm 
  \newcommand{\bp}{\begin{proof}}
  \newcommand{\bpo}{ \begin{proof}[Proof Outline] }
  \newcommand{\ep}{\end{proof}}       
\else 
  \newcommand{\bp}{\begin{IEEEproof}}     
  \newcommand{\bpo}{ \begin{IEEEproof}[Proof Outline] }
  \newcommand{\ep}{\end{IEEEproof}}       
\fi

\newcommand{\be}{\begin{equation}}
\newcommand{\ee}{\end{equation}}

\newcommand\MyIncludeGraphics[2][]{
    \IfFileExists{#2}{%
        \includegraphics[#1]{#2}%
    }{%
        \missingfigure[figwidth=7.0cm]{Missing #2}%
    }%
}%

\newcommand{\T}[1]{\smallskip\noindent\textbf{#1}} 
\newcommand{\U}[1]{\smallskip\noindent\emph{#1}} 

\newcommand{\para}[1]{\left( #1 \right)}        
\newcommand{\brac}[1]{\left\{ #1 \right\}}
\newcommand{\set}[1]{\left\{#1\right\}}         

\newcommand{\name}{Anchor\-Hash\xspace}

\newcommand{\newVar}[2]{\newcommand{#1}{\ensuremath{#2}\xspace}}
  \newVar{\server}{S}
  \newVar{\client}{C}
  \newVar{\rclient}{R_c}
  \newVar{\rserver}{R_s}

\providecommand{\ie}{\emph{i.e.,} }
\providecommand{\eg}{\emph{e.g.,} }

\newcommand{\vx}{\checkmark\kern-1.1ex\raisebox{.7ex}{\rotatebox[origin=c]{125}{--}}}

\begin{document}

\title{AnchorHash: A Scalable Consistent Hash}

\ifblind
	\author{} 
\else
  \ifacm 
	\ifacmart
    \newcommand{\aut}[2]{#1\texorpdfstring{$^{#2}$}{(#2)}} 
\author{
  \aut{Grad Student}{1},
  \aut{Isaac Keslassy}{1,2}
}%
      \affiliation{
$^1$ \emph{Technion} \quad 
$^2$ \emph{VMware} }
          \renewcommand{\shortauthors}{S. Chole \emph{et al.}}     
 \acmConference[CoNEXT'17]{ACM CoNEXT}{December 2017}{Seoul, South Korea}
\settopmatter{printfolios=true,printacmref=false} 
   	  \setcopyright{none}
	\else 
    	\author{List of authors}
     \fi
        
        
  \else 
     \Conf{
     
     \author{{Gal Mendelson,
        Shay Vargaftik,
        Katherine Barabash,
        Dean Lorenz,
        Isaac Keslassy, 
        and Ariel Orda 
\thanks{ Gal Mendelson is with the Technion and IBM Research, e-mail: galmendelson@gmail.com. Shay Vargaftik is with VMware Research, e-mail: shayv@vmware.com. Katherine Barabash and Dean Lorenz are with IBM Research, e-mail: \{dean@il, kathy@il\}.ibm.com. Isaac Keslassy and Ariel Orda are with the Technion, e-mail: \{isaac@ee, ariel@ee\}.technion.ac.il.}}

          
          }
      }
      \TR{\author{\IEEEauthorblockN{\large Grad Student and Isaac Keslassy}}}
      
  \fi 
\fi 

\OnlyTR{\markboth{Technical Report TR16-01, Technion, Israel}{}}

\ifacmart
\else
	\maketitle
\fi

\ifacm 
    \sloppypar
\else 
\fi

\newcommand{\noi}{\noindent}
\newcommand{\Z}{{\mathbb Z}}


\begin{abstract}
Consistent hashing  
is a central building block in many networking applications, such as maintaining connection affinity of TCP flows. However, current consistent hashing solutions do not ensure full consistency under arbitrary changes or scale poorly in terms of memory footprint, update time and key lookup complexity. 

We present AnchorHash, a scalable and fully-consistent hashing algorithm. \name achieves high key lookup rate, low memory footprint and low update time. We formally establish its strong theoretical guarantees, and present an advanced implementation with a memory footprint of only a few bytes per resource. Moreover, evaluations indicate that AnchorHash scales on a single core to 100 million resources while still achieving a key lookup rate of more than 15 million keys per second.
 
\end{abstract}

\begin{IEEEkeywords}
Consistent hashing, load balancing.
\end{IEEEkeywords}


\ifacmart
	\maketitle
\else
\fi

\TR{
 \let\thefootnote\relax\footnotetext{
  \Journal{
    \textbf{[Add this 1st paragraph in 1st journal submission only:]} This paper was presented in part at IEEE Infocom '15, Charleston SC, April 2015. Additions to the conference version include new theorems, complete proofs that were previously omitted for space reasons, and additional simulation results.
  }

 G. Student and I. Keslassy are with the Department of Electrical Engineering, Technion, Israel (e-mails: \{grad@tx, isaac@ee\}.technion.ac.il).
 }
}


\section{Introduction}
\label{sec_introduction}
\noindent\textbf{Background.} Consistent hashing (CH) aims at 
mapping the identifiers ({keys}) of incoming objects  
into a set of resources, while achieving (1)~\emph{minimal disruption}, \ie minimum mapping changes as resources are arbitrarily removed or added, and 
(2)~\emph{balance}, \ie even spreading of the keys across resources 
such that no resource is overloaded. 

CH is a central building block in many networking applications, such as datacenter load balancing, distributed hash tables, and distributed storage\cite{maglev,goel2017alternate, chord, halperin2014demonstration, fonseca2005beacon,dynamo,kademlia}. For instance, it is used by L4 load-balancers to evenly forward incoming packets to servers, while  maintaining the affinity of TCP connections while servers are removed or added. However, as we next describe, despite recent advances~\cite{maglev,beamer,faild}, current CH solutions do not ensure full consistency under arbitrary changes or scale poorly in terms of memory footprint, update time and key lookup complexity. 

\T{Related Work.} 
Consistent hashing was first introduced  in the context of caching using the Ring algorithm (also called Consistent Hashing)~\cite{karger,webcaching}. Several variations of the traditional Ring algorithm have been suggested in the literature to improve balance, \eg \cite{multiprobe,perfect}. Such Ring-based solutions face scalability issues, since they require a significant memory footprint and an increasing key lookup complexity.

Another well-known CH algorithm is Highest-Random-Weight (HRW)~\cite{HRW}, also designed with the goal of increasing cache hit rates. It was later applied in the design of a location service for wireless networks 
\cite{wireless-hrw} as well as in data storage systems~\cite{coblitz}. While HRW offers good balance and small memory footprint, its computational complexity is prohibitive. 

To achieve high key-lookup rate, MaglevHash \cite{maglev} and similar techniques (\eg \cite{beamer,faild})  rely on large memory tables.
These solutions
sacrifice full consistency, memory footprint, and update time upon resource additions and removals. 

Several additional algorithms are designed for special cases where resources cannot be removed or added arbitrarily.  For example, Jump 
Consistent Hash 
\cite{jump} assumes that resources can only be added or removed in a specific order. Two additional approaches that do not support resource additions are considered in~\cite{globecom}. The second approach shares some design features with our algorithm, but its implementation cannot scale due to a large memory footprint. 

\T{\name.} In this paper, we present \emph{\name}, a new hashing technique that guarantees minimal disruption, balance, high lookup rate, low memory footprint, and fast update time after resource additions and removals. Table~\ref{tab:main_comp} shows how \name is the only algorithm 
to achieve these goals at once. As opposed to the other algorithms, \name's decisions depend on past events in the system. 

We first introduce \name, which hashes the incoming object's key into successively smaller sets of resources until eventually obtaining its unique mapped resource. We show how \name stays consistent under arbitrary resource removals and additions by keeping some history (\se{sec_anchorhash}). 

Then, we formally prove that \name is consistent, \ie guarantees minimal disruption and balance. We further prove that the average number of required hash computations in a key lookup depends only on the fraction of randomly failed resources and not on their absolute number. This allows for a very high key lookup rate at scale. We prove that even under extreme failure conditions, where $50\%$  of resources are removed in an adversarial manner, a key lookup by \name still requires less than $\approx1.69$ hash computations on average and admits a low standard deviation of less than $\approx 0.83$ (\se{sec:properties}).

\begin{table}[!t]
\centering  \scriptsize  \ra{0.9}
\begin{tabular}
{p{1.3cm} l >{\centering}p{.65cm} >{\centering}p{.65cm} >{\centering}p{.65cm} >{\centering}p{.65cm}l} 
\toprule
& & HRW \cite{HRW} & Ring \cite{karger} & Maglev\-Hash~\cite{maglev} & \textbf{\name} &\\
\midrule
\textbf{Consistency} &
\textbf{Minimal disrupt.} & \checkmark & \checkmark & $\times$ & \checkmark &\\
& \textbf{Balance} & \checkmark & \vx & \vx & \checkmark &\\
\midrule
\textbf{Scalability} &
\textbf{Lookup rate} & $\times$ & $\times$ & \checkmark & \checkmark &\\
& \textbf{Memory} & \checkmark & \vx & \vx & \checkmark &\\
& \textbf{Update time} & \checkmark & \vx & $\times$ & \checkmark &\\
\midrule
\textbf{Statelessness} & & \checkmark & \checkmark & \checkmark & $\times$ &\\  
\bottomrule 
\end{tabular}
\caption{Comparison of \name and common CH algorithms. Existing algorithms sacrifice full consistency and/or scalability, while \name aims at providing both. As we later show, \name leverages state information to achieve its properties and therefore is not stateless.}
\label{tab:main_comp}
\end{table}

Next, we focus on implementing \name. Using several successive improvements in the data representation structures, we show how \name can be reduced to an $O(1)$ memory footprint per resource, at the cost of a slight increase in complexity (\se{sec:anchor:implementation}). 

We then evaluate \name as well as HRW~\cite{HRW}, Ring~\cite{karger} and MaglevHash~\cite{maglev}, using the criteria in Table~\ref{tab:main_comp}. \name and MaglevHash are the only algorithms that achieve high key lookup rate at scale, but MaglevHash sacrifices its consistency, and also requires a high memory footprint and a prohibitive update time. On the other hand, \name achieves low memory footprint and a negligible update time. In fact, we find that \name scales on a single core to 100 million resources while achieving a key lookup rate of more than 15 million keys per second (\se{sec:eval}). Finally, the code for AnchorHash appears in \cite{anchorcode}.


\section{Preliminaries} 
\label{sec_preliminaries}

We wish to map object keys to resources.
Let $\mathcal{U}$ denote the set of keys, and let $\mathcal{S}$ denote the current set of resources. For example, in the context of datacenter L4 load-balancing, keys may correspond to packet 5-tuples and resources to servers.

\T{Mapping keys to resources.} As \fig{fig:scheme} illustrates, we use \emph{indirection} by first mapping keys to \emph{buckets}, then {buckets} to resources. Specifically, current existing resources, \ie members of $\mathcal{S}$, are assigned to buckets in a one-to-one correspondence.
Buckets belong to a set denoted by $\mathcal{A}$, whose size corresponds to the number of available resources in the system, \eg the number of servers in the load balanced cluster. Let $\mathcal{W} \subseteq \mathcal{A}$ denote the subset of buckets that are currently assigned to resources, which we call the \emph{working set}. We refer to buckets in $\mathcal{W}$ as \emph{working} buckets. Note that the set $\mathcal{A}$ of all possible buckets is fixed, while its subset  $\mathcal{W}$ changes upon resource removals/additions, \eg due to server failures or maintenance operations.
Thus the mapping can be decomposed into two parts: \\
\emph{(i) Keys to buckets.} A key is first mapped to a bucket in $\mathcal{W}$.\\
\emph{(ii) Buckets to resources.} The corresponding resource in $\mathcal{S}$ is deduced from the bucket using the indirection.

\begin{figure}[!t]
\centering
\MyIncludeGraphics[width=0.99\linewidth]{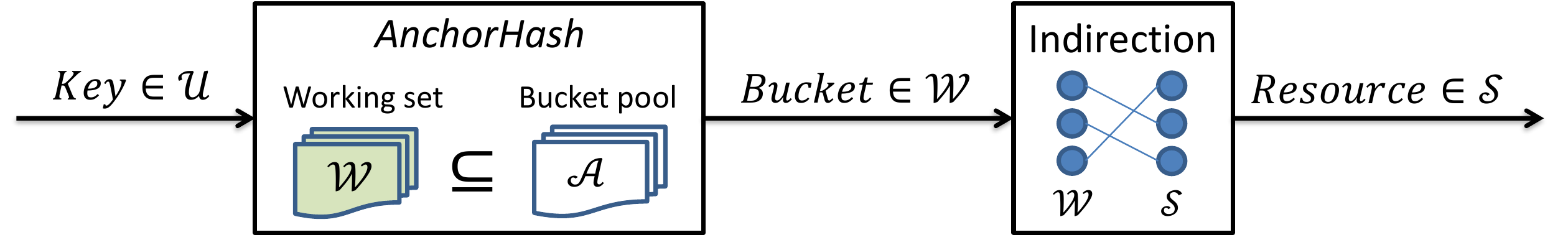}
\parbox{0.95\columnwidth}{%
\caption{\name uses indirection in order to compute the key-to-resource mapping. It first sets a bijective mapping between buckets and resources (right side), and then computes for each incoming key a key-to-bucket mapping (left side). 
}
\label{fig:scheme}}
\end{figure}

\T{Resource removal and addition.} 
Resources can be added and removed arbitrarily. Upon a removal, the corresponding (bucket,resource) pair is removed from the indirection, and the bucket is removed from $\mathcal{W}$. When a resource is added, it is assigned a bucket in $\mathcal{A}{\setminus}\mathcal{W}$, the bucket is added to $\mathcal{W}$ and the pair (bucket,resource) is added to the indirection. 

Note that a resource removal uniquely determines the bucket to remove from $\mathcal{W}$. However, when a resource is added, due to the indirection, any bucket in $\mathcal{A}{\setminus}\mathcal{W}$ can be added. This property is one of the building blocks we use to construct \name. 

The rest of this section is devoted to the first part of mapping keys to buckets, since the second indirection-based part is straightforward. We henceforth refer to adding/removing a resource as adding/removing a bucket. 


\T{Goals.} We start by formally defining our goals.
We seek a consistent hash algorithm that maps keys to buckets and satisfies the following joint objectives of \emph{minimal disruption} and \emph{balance}:

\begin{definition}[Minimal disruption] \label{def:md}
A hash algorithm achieves  \emph{minimal disruption}  iff \\
\emph{(i)} Upon the addition of a bucket $b \in \mathcal{A} {\setminus} \mathcal{W}$ to $\mathcal{W}$, keys either maintain their mapping or are remapped to $b$. \\
\emph{(ii)} Upon the removal of a bucket $b \in \mathcal{W}$, keys that were not mapped to $b$ keep their mapping, and keys that were mapped to $b$ are remapped to members of $\mathcal{W} {\setminus} \{b\}$.
\end{definition}



\U{Example} Consider a hash algorithm
$\textsc{Hash}(\mathcal{W}, k) =  \mathcal{W}[h]$, where 
\begin{equation*}
h = 
\begin{cases}
h_1(k) \equiv \textit{hash}(k) \mod \vert \mathcal{A} \vert &\quad\text{if $h_1(k) < \vert \mathcal{W} \vert $,} \\
h_2(k) \equiv \textit{hash'}(k) \mod \vert \mathcal{W} \vert  &\quad\text{otherwise,}
\end{cases}
\end{equation*}
and $\mathcal{W} \subset \mathcal{A}$. Consider the returned hash upon removal of the last bucket $b$ from $\mathcal{W}$. If $\textsc{Hash}(\mathcal{W}, k) = \mathcal{W}[h_1(k)]$ then, since $h_1(k)$ does not depend on $\mathcal{W}$, $\textsc{Hash}(\mathcal{W}\setminus\{b\}, k) = \textsc{Hash}(\mathcal{W}, k)$, unless $\mathcal{W}[h] = b$. Minimal disruption holds for this subset of keys. 
However, if $\textsc{Hash}(\mathcal{W}, k) = \mathcal{W}[h_2(k)]$ then
most keys would get a different result, including keys that were not previously mapped to $b$. Minimal disruption does not hold in this case.

\begin{definition}[Balance] \label{def:balance}
Let $k \in \mathcal{U}$ be a key. A hash algorithm achieves \emph{balance} iff $k$ has an equal probability of being mapped to each bucket in $\mathcal{W}$. 
\end{definition}

\begin{definition}[Consistency] \label{def:c}
We define a hash algorithm as \emph{consistent} iff it achieves both \emph{minimal disruption} and \emph{balance}. 
\end{definition}

\name uniformly hashes keys to bucket sets using hash functions. Accordingly, for our theoretical exposition, we make the following standard assumption (\eg \cite{gonnet1991handbook,knuth1998sorting}). 

\T{Uniform hashing assumption.} For every subset $\mathcal{B} \subseteq \mathcal{A}$, let 
$H_{\mathcal{B}}:\mathcal{U} \rightarrow \mathcal{B}$ be a hash function which maps keys to $\mathcal{B}$ such that: (1) $\forall k \in \mathcal{U}$ we have that $H_{\mathcal{B}}(k)$ is a uniform random variable on $\mathcal{B}$, and (2) for any sequence of distinct subsets ${\mathcal{B}}_1,{\mathcal{B}}_2,\dots$ the random variables $H_{{\mathcal{B}}_1}(k),H_{{\mathcal{B}}_2}(k),\ldots$ are independent. In practice, this can be approximated by introducing a random \emph{seed} into a hash function, \eg $H_{\mathcal{B},Seed}$.

\section{AnchorHash}
\label{sec_anchorhash}

\subsection{\name principles} 
We now explain how \name maps keys to buckets. We start with an initial working set, and then discuss how buckets are removed and added.

\T{Initial mapping.} Suppose we begin with a working set $\mathcal{W}$. We use the hash function $H_{\mathcal{W}}$ to map keys to 
$\mathcal{W}$.
By the uniform hashing assumption, for any key $k$, each member of $\mathcal{W}$ has an equal probability to be chosen, thus achieving balance (Def.~\ref{def:balance}).

\T{Bucket removal.}
Now, suppose that we want to remove a bucket $b \in \mathcal{W}$. 
If we use the new hash function $H_{\mathcal{W} \setminus b}$ to map keys to buckets, keys that were mapped to members of $\mathcal{W} \setminus b$ by $H_{\mathcal{W}}$ might be remapped, and the minimal disruption property will not hold.  

To address this issue, the key idea in \name is to keep using $H_\mathcal{W}(k)$ as long as $H_\mathcal{W}(k)\neq b$, and otherwise  rehash the key to $\mathcal{W} \setminus b$ using $H_{\mathcal{W} \setminus b}(k)$. For instance, assume that the initial working set is $\mathcal{W}=\brac{0,\dots,6}$. Then we are hashing any key $k$ using $H_{\brac{0,\dots,6}}(k)$. Assume now that bucket $6$ is removed. Then we continue to first hash any key $k$ using $H_{\brac{0,\dots,6}}(k)$. If it hits a bucket in $\brac{0,\dots,5}$, we are done. Otherwise, we rehash the key using $H_{\brac{0,\dots,5}}(k)$, with the result guaranteed to be a working bucket.  

This approach preserves the consistency of the algorithm, as we later formally prove. First, only keys that were mapped to $b$ are remapped, thus minimal disruption is achieved. Second, by the uniform hashing assumption, keys that did not initially hit $b$ are spread uniformly over $\mathcal{W} \setminus b$, and the same is true for the keys that initially hit $b$ and are rehashed. Therefore, balance is also achieved. 

When several buckets are removed, we repeat this procedure iteratively until hitting a bucket in the working set. To simplify the notation, we denote by $\mathcal{W}_b$ the working set right after the removal of a bucket $b$. 

\U{Example.} \fig{fig:main_example} illustrates this procedure with an initial working set $\mathcal{W}{=}\{0,1,2,3,4,5,6\}$ and buckets $6$, $5$ and $1$ removed consecutively.
\fig{fig:anchor_and_example_1} shows a simple example of a key that is immediately hashed to a bucket in the working set. \fig{fig:anchor_and_example_2} shows a more complex example in which the key is repeatedly hashed to decreasing subsets until reaching a bucket in the working set.


\begin{figure}[!t]
\centering\footnotesize
\begin{subfigure}[c]{\columnwidth} \centering
\begin{tabular}{@{}l@{\hspace{-.03\columnwidth}}C{0.13\columnwidth}@{}C{0.13\columnwidth}@{}C{0.13\columnwidth}@{}C{0.13\columnwidth}@{}C{0.13\columnwidth}@{}C{0.14\columnwidth}@{}C{0.17\columnwidth}@{}}
\toprule
                                          & 0           & 1             & 2           & 3           & 4           & 5           & 6 \makebox[0.04\columnwidth]{}           \\  
\midrule 
$\mathcal{W}$       & \checkmark  & \checkmark    & \checkmark  &  \checkmark & \checkmark  & \checkmark  & \checkmark  \makebox[0.04\columnwidth]{} \\ 
$\mathcal{W}_b$     &             &               &             &             &             &             &             \\ 
\bottomrule 
\end{tabular}\\[-1ex]
\caption{\footnotesize Initial working set $\mathcal{W}{=}\{0,1,2,3,4,5,6\}$.}
\label{fig:main_example:a}
\end{subfigure}\\[1ex]

\begin{subfigure}[c]{\columnwidth} \centering
\begin{tabular}{@{}l@{\hspace{-.03\columnwidth}}C{0.13\columnwidth}@{}C{0.13\columnwidth}@{}C{0.13\columnwidth}@{}C{0.13\columnwidth}@{}C{0.13\columnwidth}@{}C{0.14\columnwidth}@{}C{0.17\columnwidth}@{}}
\toprule
                                          & 0           & 1             & 2           & 3           & 4           & 5           & 6 \makebox[0.04\columnwidth]{}          \\  
\midrule 
$\mathcal{W}$       & \checkmark  & \checkmark    & \checkmark  &  \checkmark & \checkmark  & \checkmark  & $\times$ \makebox[0.04\columnwidth]{} \\ 
$\mathcal{W}_b$     &             &               &             &             &             &             &  $\scriptstyle\set{0,1,2,3,4,5}$           \\ 
\bottomrule 
\end{tabular}\\[-1ex]
\caption{\footnotesize Removing bucket $6$ with $\mathcal{W}_6{=}\{0,1,2,3,4,5\}.$}
\label{fig:main_example:b}
\end{subfigure}\\[1ex]


\begin{subfigure}[c]{\columnwidth} \centering
\begin{tabular}{@{}l@{\hspace{-.03\columnwidth}}C{0.13\columnwidth}@{}C{0.13\columnwidth}@{}C{0.13\columnwidth}@{}C{0.13\columnwidth}@{}C{0.13\columnwidth}@{}C{0.14\columnwidth}@{}C{0.17\columnwidth}@{}}
\toprule
                                          & 0           & 1             & 2           & 3           & 4           & 5           & 6 \makebox[0.04\columnwidth]{}          \\  
\midrule 
$\mathcal{W}$       & \checkmark  & \checkmark    & \checkmark  &  \checkmark & \checkmark  & $\times$  & $\times$ \makebox[0.04\columnwidth]{} \\ 
$\mathcal{W}_b$     &             &               &             &             &             &      $\scriptstyle\set{0,1,2,3,4}$          &  $\scriptstyle\set{0,1,2,3,4,5}$           \\ 
\bottomrule 
\end{tabular}\\[-1ex]
\caption{\footnotesize Removing bucket $5$ with $\mathcal{W}_5{=}\{0,1,2,3,4\}.$}
\label{fig:main_example:c}
\end{subfigure}\\[1ex]


\begin{subfigure}[c]{\columnwidth} \centering
\begin{tabular}[l]{@{}l@{\hspace{-.03\columnwidth}}C{0.13\columnwidth}@{}C{0.13\columnwidth}@{}C{0.13\columnwidth}@{}C{0.13\columnwidth}@{}C{0.13\columnwidth}@{}C{0.14\columnwidth}@{}C{0.17\columnwidth}@{}}
\toprule
                                          & 0           & 1             & 2           & 3           & 4           & 5           & 6 \makebox[0.04\columnwidth]{}          \\  
\midrule 
$\mathcal{W}$       & \checkmark  & $\times$    & \checkmark  &  \checkmark & \checkmark  & $\times$  & $\times$ \makebox[0.04\columnwidth]{} \\ 
$\mathcal{W}_b$     &             & $\scriptstyle\set{0,2,3,4}$ &             &             &             &      $\scriptstyle\set{0,1,2,3,4}$          &  $\scriptstyle\set{0,1,2,3,4,5}$           \\ 
\bottomrule 
\end{tabular}\\[-1ex]
\caption{\footnotesize Removing bucket $1$ with $\mathcal{W}_1{=}\{0,2,3,4\}.$}
\label{fig:main_example:d}
\end{subfigure}

\caption{Example with an initial working set $\mathcal{W}=\set{0,1,2,3,4,5,6}$ (\fig{fig:main_example:a}). Then, bucket $6$, $5$ and $1$ are removed consecutively (in Figures \ref{fig:main_example:b}, \ref{fig:main_example:c} and \ref{fig:main_example:d}, respectively).
}
\label{fig:main_example}
\end{figure}





\T{Bucket addition.}  Suppose that the last bucket that was removed was $b$, and the current working set is $\mathcal{W}$ (\ie $\mathcal{W}_b$=$\mathcal{W}$). Recall that \name may add any bucket not in $\mathcal{W}$ by virtue of the indirection. 
If we need to add a new bucket, we choose to add back bucket $b$.
More generally, upon bucket addition, \name \emph{always} adds the last removed bucket. We show in \se{sec:anchor:implementation} that this allows for an extremely efficient implementation. 
This is because by our iterative construction, adding the last removed bucket $b$ simply brings us back to the state just before $b$'s removal. Specifically, upon the addition of $b$: (1) the only remapped keys are the ones remapped to $b$ (these are the same keys that hit $b$ and were rehashed after $b$ was previously removed), and \textit{minimal disruption} holds; and (2) since \textit{balance} was achieved before $b$ was removed, it is also achieved after it is added back. We prove these claims formally in \se{sec:properties}.

\U{Example.} Consider \fig{fig:main_example:d}. If we add the last removed bucket~$1$, we simply return to the state illustrated in \fig{fig:main_example:c}. At this point, if we add the last removed bucket $5$, we simply return to the state illustrated in \fig{fig:main_example:b}, and so on.

We maintain a LIFO queue (\ie stack) for the removed buckets, denoted by $\mathcal{R}$. For example, in the state illustrated in \fig{fig:main_example:d}, $\mathcal{R}{=}\{6 \leftarrow 5 \leftarrow 1\}$. 

\T{Anchor.} By construction, $\vert \mathcal{A} \vert$ is an upper bound on the number of buckets that we allow. Therefore, in practice, we simply set the value of $\vert \mathcal{A} \vert$ to a larger value than may be needed  (\eg 2$\times$ the initial system size) and insert the unused buckets (\ie members of $\mathcal{A} {\setminus} \mathcal{W}$) into the stack $\mathcal{R}$. Note that this initial order within $\mathcal{R}$ may be arbitrary. We later leverage this observation to optimize implementation. Since $\mathcal{A}$ serves as the starting point of the algorithm on which everything is defined, we refer to it as the \emph{Anchor}. 

\U{Example.} Consider again \fig{fig:main_example:a}. Assume that instead of beginning our operation with $\mathcal{W}{=}\{0,1,2,3,4,5,6\}$, we would like to start our system with only $\mathcal{W}{=}\{0,1,2,3,4\}$, but want to be prepared to increase $\mathcal{W}$ to include buckets $5$ and $6$ if needed. Then, we simply start our system with $\mathcal{A}{=}\{0,1,2,3,4,5,6\}$, and initially set $\mathcal{W}{=}\{0,1,2,3,4\}$ and $\mathcal{R}{=}\{6 \leftarrow 5\}$. This precise state is illustrated in \fig{fig:main_example:c}.  

\begin{figure}[!t]
\centering

\begin{subfigure}[c]{\columnwidth} \centering
\MyIncludeGraphics[width=0.99\linewidth]{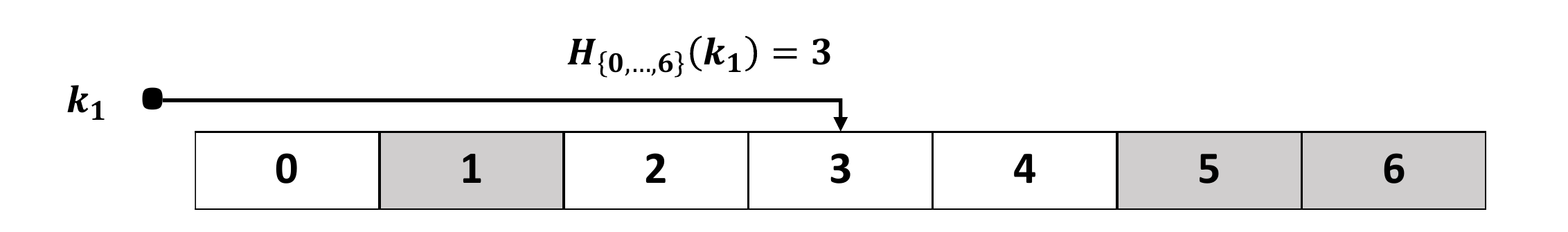}
\parbox{0.8\columnwidth}{%
\caption{\footnotesize$k_1$ is immediately hashed to bucket $3 \in \mathcal{W}$.}
\label{fig:anchor_and_example_1}}
\end{subfigure}

\begin{subfigure}[c]{\columnwidth} \centering
\MyIncludeGraphics[width=0.99\linewidth]{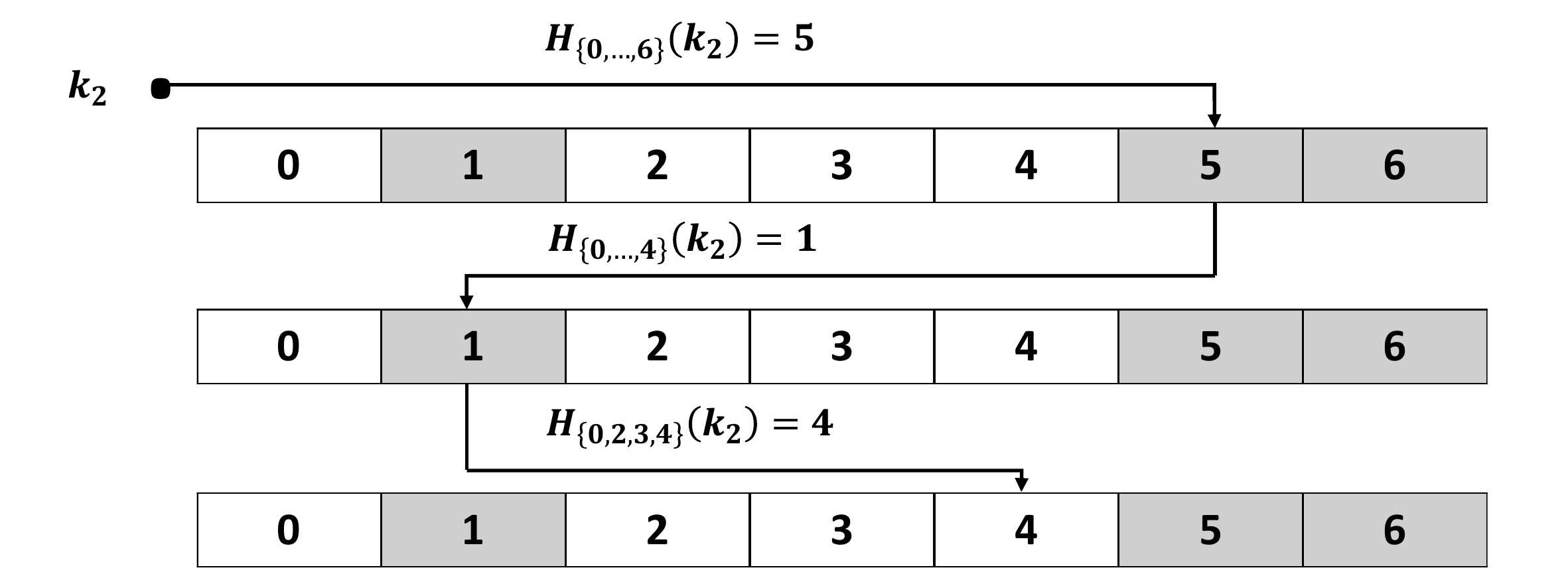}
\parbox{0.8\columnwidth}{%
\caption{\footnotesize$k_2$ is initially hashed to bucket $5$. Then, since $5 \not\in \mathcal{W}$, following \fig{fig:main_example:d}, $k_2$ is rehashed into the set $\{0,...4\}$. Assume that it is rehashed to $1 \in \mathcal{W}_5$. Since $1 \not\in \mathcal{W}$ as well, $k_2$ is rehashed again to $4 \in \mathcal{W}_1$. Since $4 \in \mathcal{W}$, the process terminates.}
\label{fig:anchor_and_example_2}}
\end{subfigure}

\caption{Example of possible key lookups in the state presented in \fig{fig:main_example:d}.} 
\label{fig:anchorhash}
\end{figure}
\subsection{\name algorithm}
\begin{algorithm}[!t]
	\footnotesize
    \begin{algorithmic}[1]
     \Function{InitAnchor}{$\mathcal{A}, \mathcal{W}$}
        \State $\mathcal{R} \gets \emptyset$
        \For{$b \in \mathcal{A} \backslash \mathcal{W}$}
        	\State $\mathcal{R}.push(b)$
        	\State $\mathcal{W}_b \gets \mathcal{A} \backslash \mathcal{R}$
        \EndFor
    \EndFunction
    
	\\\hrulefill

        \Function{GetBucket}{$k$}
        	\State $b \gets H_{\mathcal{A}}(k)$
        	\While {$b \not \in \mathcal{W}$}
        		\State $b \gets H_{\mathcal{W}_b}(k)$
        	\EndWhile
        	\State\Return $b$
        \EndFunction
	\\\hrulefill
	
        \Function{AddBucket}{ }
			\State $b \gets \mathcal{R}.pop()$
            \State $delete \,\, \mathcal{W}_b$
            \State $\mathcal{W} \gets \mathcal{W} \cup \set{b}$
            \State\Return $b$
        \EndFunction

    \\\hrulefill    

        \Function{RemoveBucket}{$b$}
        	\State $\mathcal{W} \gets \mathcal{W} \backslash \set{b}$
        	\State $\mathcal{W}_b \gets \mathcal{W}$
            \State $\mathcal{R}.push(b)$
        \EndFunction
                        
    \end{algorithmic}   
	\normalsize
    \caption{ --- AnchorHash}
    \label{alg:AnchorHash}
    
\end{algorithm}

The pseudo-code for \name is given in \al{alg:AnchorHash}.

\T{Initialization.} \textsc{InitAnchor($\mathcal{A},\mathcal{W}$)} receives as an input the Anchor $\mathcal{A}$ and the initial working set of buckets $\mathcal{W}$. We fill the stack $\mathcal{R}$ with the initially unused buckets. For each such unused bucket $b$, we remember $\mathcal{W}_b$, \ie the working set just after its removal. 

 \T{GetBucket.} \textsc{GetBucket($k$)} receives a key $k \in \mathcal{U}$ as an input and returns a working bucket $b \in \mathcal{W}$ as an output. Initially, we hash the key uniformly over the Anchor $\mathcal{A}$; then, if the calculated bucket $b$ is not a member of $\mathcal{W}$, the key is rehashed into $\mathcal{W}_b$. This process continues until hitting a working bucket. We analyze the computational complexity of this procedure in \se{sec:properties} and present empirical evaluation results in \se{sec:eval}.

\T{AddBucket.} As mentioned, when adding a bucket, we add the the last removed bucket. Accordingly, \textsc{AddBucket()} has {no input} and simply returns the added bucket. It pops the last removed bucket $b$ from $\mathcal{R}$, deletes the no-longer-needed $\mathcal{W}_b$, adds $b$ to $\mathcal{W}$ and returns $b$.

\T{RemoveBucket.} \textsc{RemoveBucket($b$)} receives as an input the bucket we want to remove, and has no return value. We simply remove $b$ from $\mathcal{W}$, record the working set just after $b$'s removal $\mathcal{W}_b$ and push $b$ to the top of $\mathcal{R}$.

\T{Indirection.} For completeness, \al{alg:AnchorHashWrapper} presents the full key-to-resource mapping based on indirection (as presented in \fig{fig:scheme}). It complements the key-to-bucket mapping of  \al{alg:AnchorHash} with a standard bucket-to-resource bijection function $M$. For simplicity, we represent this bijection using a set of coupled pairs $(b,\xi)\in M$ such that $M(b)=\xi$ and $M^{-1}(\xi)=b$. Note that such indirection is trivially implemented using a standard map with $O(1)$ operations on average for each bucket or resource lookup.


\begin{algorithm}[!t]
	\footnotesize
    \begin{algorithmic}[1]
    	\Function{InitWrapper}{$\mathcal{A},\mathcal{S}$}
          \State $M \gets \emptyset$, $\mathcal{W} \gets \emptyset$
          \For{$i \in (0,1,\ldots,|\mathcal{S}|-1)$}
              \State $M \gets M \cup \{(\mathcal{A}[i], \mathcal{S}[i])\}$
              \State $\mathcal{W} \gets \mathcal{W} \cup \set{\mathcal{A}[i]}$
          \EndFor
          \State \Call{InitAnchor}{$\mathcal{A}, \mathcal{W}$}
        \EndFunction  

	\\\hrulefill
   
        \Function{GetResource}{$k$} \Comment{Compute resource for key $k$}
			\State $b \gets$ \Call{GetBucket}{$k$}
            \State $\xi \gets M(b)$
	        \State\Return $\xi$
        \EndFunction

	\\\hrulefill

        \Function{AddResource}{$\xi$}
			\State $b \gets$ \Call{AddBucket}{ }
            \State $M \gets M \cup \set{(b,\xi)}$ 
        \EndFunction
        
    \\\hrulefill    

		\Function{RemoveResource}{$\xi$}
        	\State $b \gets M^{-1}(\xi)$ 
            \State $M \gets M \backslash \set{(b,\xi)}$
            \State \Call{RemoveBucket}{$b$}
        \EndFunction
               
    \end{algorithmic}   
	\normalsize
    \caption{ --- AnchorHash Wrapper}
    \label{alg:AnchorHashWrapper}
    
\end{algorithm}


\section{\name properties}\label{sec:properties}

In this section we first prove that \name is consistent (\ie provides minimal disruption and balance), and then analyze its complexity.

    \begin{theorem}[Minimal disruption]\label{thm:mindisp}
\name guarantees minimal disruption.
\end{theorem}

\begin{proof} 
\T{(i).} Assume a newly added bucket $b$. Consider function \textsc{GetBucket($k$)}. Then, before $b$'s addition, each $k \in \mathcal{U}$ either encountered bucket $b$ before terminating or not.  After the addition of $b$, keys that did not encounter $b$ are clearly not affected. Those that did now terminate at $b$. 

\T{(ii).} Assume a newly removed bucket $b$. Consider again function \textsc{GetBucket($k$)}. Before $b$'s removal, each $k \in \mathcal{U}$ either terminated at bucket $b$ or did not encounter it at all. After the removal of $b$, keys that did not encounter $b$ are clearly not affected. Those that did, now terminate at $H_{\mathcal{W}_b}(k)$.
\end{proof}

\begin{theorem}[Balance]
\name achieves balance.
\end{theorem}

\begin{proof} 
We prove that given a possible sequence of operations, where an operation can be either a bucket removal or a bucket addition, balance holds in the initial state and after every operation. By the definition of \textsc{GetBucket($k$)} in Algorithm~\ref{alg:AnchorHash}, its returned bucket for a specific key depends only on the stack $\mathcal{R}$ of removed buckets. Denote by $\mathcal{R}_j$ the stack after operation $j$ and denote $\mathcal{R}_0=\emptyset$. We refer to the initial state $\mathcal{R}_0$ as the state after operation number 0. Our proof is by induction on the number of operations.

\T{Basis: initial state.} In this case, since $\mathcal{R}_0 = \emptyset$ we have $\mathcal{W}{=}\mathcal{A}$. According to the uniform hashing assumption, for every $k$, $H_{\mathcal{A}}(k)$ is a uniform random variable over $\mathcal{W}$.

\T{Induction hypothesis.} Assume that balance holds after every operation $j \leq i$; namely for every stack $\mathcal{R}_j$ such that $j \leq i$.

\T{Inductive step.} We now prove that balance holds after operation $i+1$, which can either be a bucket removal (unless $\vert \mathcal{W} \vert =1$) or a bucket addition (unless $\mathcal{R}_i = \emptyset$).

\textit{Bucket removal.} Consider a newly removed bucket $b$. After $b$'s removal, according to Theorem~\ref{thm:mindisp} (minimal disruption), only keys that were mapped to $b$ are remapped. These are remapped using $H_{\mathcal{W}_b}$, which, by the uniform hashing assumption, assigns each of them with an equal probability to the members of $\mathcal{W}_b$, independently from their previous mappings.

\textit{Bucket addition.}  By the definition of \textsc{AddBucket}() in Algorithm~\ref{alg:AnchorHash}, $\mathcal{R}_{i+1}$ is obtained by popping the bucket at the top of the stack $\mathcal{R}_{i}$; therefore $\mathcal{R}_{i+1} = \mathcal{R}_{j}$, for some $j<i$ (for example, if the previous operation was a bucket removal, then $\mathcal{R}_{i+1} = \mathcal{R}_{i-1}$, as an addition is an ``undo'' of removal). Thus, by the induction hypothesis, balance holds after operation $i+1$.
\end{proof}


We now turn to providing a strong theoretical guarantee on the run-time complexity of \textsc{GetBucket($k$)}, which explains why \name is able to process keys at a high rate at scale. 

\begin{theorem}[Computational complexity]\label{comp_complexity} Fix $\mathcal{A}$, $\mathcal{W}$ and $\mathcal{R}$ such that  $\vert \mathcal{A} \vert{=}a$ and $\vert \mathcal{W} \vert{=}w$.
For a key $k$, denote by $\tau$ the number of hash operations performed by \textsc{GetBucket($k$)}. 
Then:
\begin{enumerate}[leftmargin=*]
\item The average of $\tau$ 
is upper-bounded by $1+\ln(\frac{a}{w}).$
\item The standard deviation of $\tau$ is upper-bounded by  $\sqrt{\ln(\frac{a}{w})}.$
\end{enumerate}
\end{theorem}

\begin{proof} 

Once \textsc{GetBucket($k$)} is invoked, we repeatedly hash $k$ into decreasing subsets of $\mathcal{A}$ until hitting a working bucket. The number of hash operations is 1 plus the number of iterations in the while loop, which is entered only if $b = H_{\mathcal{A}}(k) \not\in \mathcal{W}$, i.e., bucket $b$ was removed. 
Consider a fixed sequence of removals $\mathcal{R} =\brac{r_{a-w-1} \leftarrow \ldots \leftarrow r_{0}}$; namely, $r_{a-w-1}$ is the first removed bucket and $r_{0}$ is the last removed bucket (i.e., $\mathcal{W} = \mathcal{W}_{r_0}$). 

Let $\tau_i$ denote the number of remaining iterations \emph{after} the loop is entered with $b=r_i$. Let $b_i$ denote $H_{\mathcal{W}_{r_i}}(k)$; if $b_i \in \mathcal{W}$ then the loop terminates and $\tau_i=1$. If $b_i \notin \mathcal{W}$, then $b_i=r_j$ for some $j<i$ ($r_i$ and all earlier removals are not in $\mathcal{W}_{r_i}$). Then, by the uniform hashing assumption, $\tau_i$ has the same distribution as $1+\tau_j$.

For $i=0$, since $\mathcal{W} = \mathcal{W}_{r_0}$, we have $b_0 \in \mathcal{W}$, thus $\tau_0 =1$. With $\tau_0,\ldots,\tau_{a-w-1}$ at hand, for ease of exposition, we also define $\tau_{a-w}=\tau$ and $b_{a-w}=b$. We use these notations and observations to derive a recursive formula and find a closed-form expression for the moment generating function (MGF) of $\tau$. We then use it to find the first two moments of $\tau$. 

For $0 \leq i \leq a-w$, define
\begin{equation}\label{eq:performance:def}
\phi_i(s) = \mathbb{E}[ e^{s\tau_i}].
\end{equation}
Then, by the law of total expectation, for $i>0$,
\begin{align}\label{eq:expanding 1}
\phi_i(s)&=\mathbb{P}(b_i {\in} \mathcal{W})\mathbb{E}[ e^{s\tau_i} {\mid} b_i {\in} \mathcal{W}]+ \sum_{j=0}^{i-1}\mathbb{P}(b_i{=}r_j)\mathbb{E}[ e^{s\tau_i} {\mid} b_i{=}r_j]\cr
&=\frac{w}{w+i}\mathbb{E}[ e^{s\tau_i} {\mid} b_i {\in} \mathcal{W}]+ \frac{1}{w+i}\sum_{j=0}^{i-1}\mathbb{E}[ e^{s\tau_i} {\mid} b_i{=}r_j].
\end{align}
First, if $b_i\in \mathcal{W}$, the loop terminates after a single hash calculation, \ie $\tau_i=1$. Thus
\begin{equation}\label{eq:expanding 2}
\mathbb{E}[ e^{s\tau_i} {\mid} b_i {\in} \mathcal{W}]=e^s.
\end{equation}
Second, recall that the distribution of $\tau_i$ conditioned on $b_i=r_j$ follows the same distribution as $1+\tau_{j}$. Therefore 
\begin{equation}\label{eq:expanding 3}
\begin{split}
\mathbb{E}[ e^{s\tau_i} \mid b_i=r_j] 
&= \mathbb{E}[ e^{s(1+\tau_{j})}]\\
& =  e^s\mathbb{E}[ e^{s\tau_{j}}]=e^s\phi_{j}(s),
\end{split}
\end{equation}
Substituting \eq{eq:expanding 2} and \eq{eq:expanding 3} in \eq{eq:expanding 1} yields
\begin{equation}\label{eq:performance:rec}
\phi_i(s) = \frac{w\cdot e^s}{w+i} + \frac{e^s}{w+i} \cdot \sum_{j=0}^{i-1}\phi_{j}(s).
\end{equation}

Now that we have a recursive formula for $\phi_i(s)$, we are able to calculate its closed-form expression. 
For $i>0$, rearranging \eq{eq:performance:rec} yields 
\begin{equation}\label{eq:performance:rec:a}
\begin{split}
\frac{w+i}{e^s} \cdot \phi_i(s) & = w+ \sum_{j=0}^{i-1}\phi_{j}(s) \\ 
&= w+ \sum_{j=0}^{i-2}\phi_{j}(s) + \phi_{i-1}(s) \\
& = \frac{w+i-1}{e^s} \cdot \phi_{i-1}(s) + \phi_{i-1}(s) \\
&= \frac{w+i-1+e^s}{e^s} \cdot \phi_{i-1}(s),
\end{split}
\end{equation}
where the third equality is derived by using the first equality but with $i-1$ instead of $i$.
Therefore, 
\begin{equation}\label{eq:performance:rec:d}
\phi_i(s) = \phi_{i-1}(s) \cdot \frac{w+i-1+e^s}{w+i}.
\end{equation}
Now, using \eq{eq:performance:rec:d} and the stopping condition $\phi_0(s)=e^s$, we obtain for $i>0$,
\begin{equation}\label{eq:mgf_indep}
\phi_i(s)=e^s\prod_{j=1}^{i} \Big(\frac{w+j-1+e^s}{w+j}\Big),
\end{equation} 
Taking the logarithm and then differentiating with respect to $s$ yields
\begin{equation}\label{eq:performance:derivative:a}
\frac{\phi'_{i}(s)}{\phi_{i}(s)}=1+\sum_{j=1}^{i}\Big(\frac{e^s}{w+j-1+e^s}\Big).
\end{equation}
By \eq{eq:performance:def}, $\phi_i(0)=1$. Hence, substituting $s=0$ in \eq{eq:performance:derivative:a} yields
\begin{equation}
\mathbb{E}[\tau_i] = \phi'_{i}(0)= 1 + \sum_{j=1}^{i}\frac{1}{w+j},
\end{equation}
and therefore
\begin{align*}
\mathbb{E}[\tau] = \mathbb{E}[\tau_{a-w}] &= \phi'_{a-w}(0) = 1 + \sum_{j=1}^{a-w}\frac{1}{w+j}\cr
&\leq 1+\int_{w}^{a}\frac{1}{x}dx=1+\ln\Big(\frac{a}{w}\Big).
\end{align*}
Now, to obtain the bound on the standard deviation, we take the derivative with respect to $s$ in \eq{eq:performance:derivative:a} and obtain
\begin{equation*}\label{eq:second derivative}
\frac{\phi''_{i}(s)\phi_{i}(s) -(\phi'_{i}(s))^2}{(\phi_{i}(s))^2}=   
\sum_{j=1}^{i}\Big(\frac{e^s(w+j-1+e^s)-e^{2s}}{(w+j-1+e^s)^2}\Big).
\end{equation*}
Setting $i=a-w$, $s=0$ and using $\phi_{a-w}(0)=1$ yields
\begin{align*}
\mbox{Var}(\tau) &= \mbox{Var}(\tau_{a-w}) = \phi''_{a-w}(0)-(\phi'_{a-w}(0))^2 \cr
& = \sum_{j=1}^{a-w}\frac{w+j-1}{(w+j)^2} \le \sum_{j=1}^{a-w}\frac{1}{w+j} \le \ln(\frac{a}{w}).
\end{align*}
Thus the standard deviation is upper bounded by $\sqrt{\ln(\frac{a}{w})}$. Note that the bounds do not depend on the removal sequence we fixed. This concludes the proof.
\end{proof}


\section{\name implementation} \label{sec:anchor:implementation}

Equation \eqref{eq:mgf_indep}, from which the result of Theorem \ref{comp_complexity} is derived,  implies that the distribution of the number of iterations (and by that also the number of hash operations) of \textsc{GetBucket($k$)} is \emph{independent} of the implementation. However, the implementation does determine the amount of used memory and how many calculations and memory accesses are performed during each iteration of \textsc{GetBucket($k$)}. 

Specifically, each such iteration requires choosing a bucket uniformly at random from a known set (\ie $b \gets H_{\mathcal{W}_b}(k)$) and checking if this bucket is working (\ie $b \not \in \mathcal{W}$). The most challenging part is finding an efficient way to hold these different sets (\ie $\{\mathcal{W}_b \,\big|\, b \in \mathcal{R}\}$). 

In the following, we first describe in detail the different components of \name implementation. Then, we present three distinct implementations of holding $\{\mathcal{W}_b \,\big|\, b \in \mathcal{R}\}$ that achieve different memory-computation complexity trade-offs which are summarized in Table \ref{tab:anchor_evo}.

\T{Anchor representation.} We use an integer array $A$ of size $a$ to represent the Anchor. Each bucket ${b \in \set{0,1,\ldots,a-1}}$ is represented by $A[b]$ that either equals $0$ if $b$ is a working bucket (\ie $A[b]=0$ if $b \in \mathcal{W}$), or else equals the size of the working set just after its removal (\ie $A[b]=\vert\mathcal{W}_b\vert$ if $b \in \mathcal{R}$). 

\U{Example.} Considering again the example in \fig{fig:main_example:d}, we have
$$
\overset{^{\footnotesize\mbox{$b$:}}}{\mbox{$A[b]$:}} \quad \overset{0}{\fbox{0}} \sep \overset{1}{\fbox{4}} \sep \overset{2}{\fbox{0}} \sep \overset{3}{\fbox{0}} \sep \overset{4}{\fbox{0}} \sep \overset{5}{\fbox{5}} \sep \overset{6}{\fbox{6}}
$$ 
By examining this array we can determine that buckets $0$, $2$, $3$, and $4$ are working, and buckets $1,5,6$ are removed, with $\vert\mathcal{W}_1\vert=4$, $\vert\mathcal{W}_5\vert=5$ and $\vert\mathcal{W}_6\vert=6$. 

\T{Hashing.} Denote $h_b(k) \equiv \mbox{hash}(b,k) \bmod A[b]$. We are using $b$ as "salt" in the hash function to make sure $\{h_b()\}_{b=0,1,\ldots}$ are independent (as assumed by the uniform hashing assumption). To implement $\mbox{hash}(b,k)$ efficiently, recent software-based solutions such as xxHash~\cite{xxhash_small_keys} and hardware-supported hashing such as crc32~\cite{singhal2008inside} can be used.

\T{Removed buckets.} 
\name saves the removed buckets in a LIFO order for possible future bucket additions. Accordingly, we use an efficient implementation of a stack data structure $R$ to hold the removed buckets. 

\U{Example.} In the example of \fig{fig:main_example:d}, $R$ looks like: 
$$\framebox(13,13){6} \framebox(13,13){5} \framebox(13,13){1} \framebox(13,13){ } \framebox(13,13){ } \framebox(13,13){ } \, \colorbox{white}{\makebox(11,11){$\LARGE\boldsymbol\leftrightarrows$}}$$

\T{Decreasing subsets.} For each removed bucket $b$, we need an efficient way of representing $\mathcal{W}_b$ and calculating $H_{\mathcal{W}_b}(k)$. 
For clarity, we tackle this challenge in stages: we begin with a \emph{naive} implementation, which we successively improve to implementations with a \emph{partial} then \emph{minimal} memory usage.

\subsection{Naive implementation} 

A naive approach to representing $\set{\mathcal{W}_b \,|\, b \in \mathcal{R}}$ is using a key-value store, $\mbox{KV}$,  that holds the pairs ${\{(b,\mathcal{W}_b) \, | \, b \in \mathcal{R}\}}$, where the \emph{key} is a removed bucket $b$ and the \emph{value} is $\mathcal{W}_b$, stored in $\mbox{KV}[b]$ as an array. 
This way, implementing $H_{\mathcal{W}_b}(k)$ simply translates to $H_{\mathcal{W}_b}(k) \equiv \mbox{KV}[b][h_b(k)]$. 

\U{Example.} In the example in \fig{fig:main_example:d}, $\mbox{KV}[1]$ looks like:
\begin{equation*}
\mathcal{W}_1 = \{ 0,2,3,4 \}, \qquad \overset{^{\footnotesize\mbox{$h_1(k)$:}}}{\mbox{KV}[1]}: \quad \overset{0}{\fbox{0}} \sep \overset{1}{\fbox{2}} \sep \overset{2}{\fbox{3}}  \sep \overset{3}{\fbox{4}}
\end{equation*}

Unfortunately, albeit simple, this approach is not scalable, as it requires to maintain an array of size $|\mathcal{W}_b|$ for each removed bucket $b$, incurring a large memory footprint of $\Theta(|\mathcal{A}|+|\mathcal{A}||\mathcal{R}|)$. That is, arrays of sizes $|\mathcal{A}|-1,|\mathcal{A}|-2,\ldots,|\mathcal{A}|-|\mathcal{R}|$ are maintained for the members of $|\mathcal{R}|$.
Since we also use an array of size $|\mathcal{A}|$ for the Anchor representation, the total memory footprint is given by  
\begin{align*}
 |\mathcal{A}| + \sum_{i=1}^{|\mathcal{R}|} \big(|\mathcal{A}|-i\big)&=|\mathcal{A}|+|\mathcal{A}||\mathcal{R}|-0.5|\mathcal{R}|(|\mathcal{R}|+1)  \cr
 &\geq |\mathcal{A}|+|\mathcal{A}||\mathcal{R}|-0.5|\mathcal{A}|(|\mathcal{R}|+1)\cr
 &=0.5(|\mathcal{A}|+|\mathcal{A}||\mathcal{R}|),
\end{align*}
and since 
\begin{align*}
 |\mathcal{A}| + \sum_{i=1}^{|\mathcal{R}|} \big(|\mathcal{A}|-i\big)&=|\mathcal{A}|+|\mathcal{A}||\mathcal{R}|-0.5|\mathcal{R}|(|\mathcal{R}|+1)  \cr
 &\le |\mathcal{A}|+|\mathcal{A}||\mathcal{R}|,
\end{align*}
the total memory footprint of the naive implementation is $\Theta(|\mathcal{A}|+|\mathcal{A}||\mathcal{R}|)$. Upon a bucket addition, the complexity accounts for $\Theta(|\mathcal{W}|)$ (\ie adding $\mathcal{W}_b$ to $\mbox{KV}$).

\subsection{Reduced-memory implementation}

\T{Non-fixed points.} Consider again the naive implementation. Recall that all of the theoretical properties of \name are \emph{independent} of the exact bucket order within the sets $\{ \mathcal{W}_b \,|\, b \in \mathcal{R}\}$. Also, for any two \emph{consecutively} removed buckets $b_1$ and $b_2$, the sets $\mathcal{W}_{b_1}$ and $\mathcal{W}_{b_2}$ only differ by a single bucket.

We want to leverage these properties to reduce the memory footprint of \name and accelerate its performance. Accordingly, we seek to \emph{minimize the number of non-fixed point entries} in the members of $\{ \mathcal{W}_b \,|\, b \in \mathcal{R}\}$, which we define as entries that respect $\mathcal{W}_b[h] \neq h$. This way we do not need to remember the full arrays, but only the \emph{difference} between the initial order of buckets and each member of $\set{\mathcal{W}_b \,|\, b \in \mathcal{R}}$, \ie the \emph{non-fixed points}.

\U{Example.} Recall the example in \fig{fig:main_example:d}. 
In this example, the naive approach holds three arrays: $\mbox{KV}[6]$, $\mbox{KV}[5]$, and $\mbox{KV}[1]$. 
Our goal is to minimize the number of non-fixed point entries between the initial order of buckets $\set{0,1,2,\ldots,6}$ and the order of buckets in the members of $\{ \mathcal{W}_b \,|\, b \in \set{1,5,6}$\}. For example, to obtain the desired order for $\mathcal{W}_1=\brac{0,2,3,4}$ and minimize the difference with $ \footnotesize \fbox{0} \sep \fbox{{1}} \sep \fbox{2} \sep \fbox{3},$ we simply use $ \footnotesize \fbox{0} \sep \fbox{\textbf{4}} \sep \fbox{2} \sep \fbox{3},$ \ie
 take bucket $4$ which is the last element in $\mbox{KV}[5]$, and put it instead of the removed bucket $1$. This yields

\begin{alignat}{2}\label{eq:naive_kv_eff}
 \mathcal{A} &= \{ 0,1,2,3,4,5,6 \}, & \mbox{Init}: & \quad \overset{0}{\fbox{0}} \sep \overset{1}{\fbox{1}} \sep \overset{2}{\fbox{2}} \sep \overset{3}{\fbox{3}} \sep \overset{4}{\fbox{4}} \sep \overset{5}{\fbox{5}} \sep \overset{6}{\fbox{6}}  \nonumber \\
 \mathcal{W}_6 &= \{ 0,1,2,3,4,5 \}, & \mbox{KV}[6]: & \quad\fbox{0} \sep \fbox{1} \sep \fbox{2} \sep \fbox{3} \sep \fbox{4} \sep \fbox{5} \nonumber \\
 \mathcal{W}_5 &= \{ 0,1,2,3,4 \}, & \mbox{KV}[5]: & \quad\fbox{0} \sep \fbox{1} \sep \fbox{2} \sep \fbox{3} \sep \fbox{4} \nonumber \\
 \mathcal{W}_1 &= \{ 0,2,3,4 \}, &\mbox{KV}[1]: & \quad\fbox{0} \sep \fbox{\textbf{4}} \sep \fbox{2} \sep \fbox{3}  
\end{alignat}

Examining \eq{eq:naive_kv_eff} reveals that instead of remembering all three arrays, we can just remember that $\mbox{KV}[1][1]{=}4$ (recall that $A$ provides the length of each array).  Namely, all other elements are simply fixed points. Each time we calculate $H_{\mathcal{W}_b}(k)$, it equals $h_b(k)$ without the need to access any data structure. The only exception is when an entering key hits bucket $1$ and then hashes to $1$ again (\ie calculating $H_{\mathcal{W}_1}(k)$ yields $h_1(k)=1$). For this specific case, we need to remember that we hit bucket $\mbox{KV}[1][1]{=}4$ instead of 1.

Now, assume that in this state bucket $0$ is removed.  Similarly, the desired ordering for $\mathcal{W}_0$ is obtained by taking the last element in $\mbox{KV}[1]$, which is bucket $3$, and putting it instead of the removed bucket $0$. This yields
\begin{equation}\label{eq:naive_kv_eff_2}
\mathcal{W}_0 = \{ 2,3,4 \}, \qquad \overset{^{\footnotesize\mbox{$h_0(k)$:}}}{\mbox{KV}[0]}: \quad \overset{0}{\fbox{\textbf{3}}} \sep \overset{1}{\fbox{\textbf{4}}} \sep \overset{2}{\fbox{2}}
\end{equation}
Again, we only need to store $\mbox{KV}[0][0]{\gets}3$ and $\mbox{KV}[0][1]{\gets}4$, since working bucket $2$ is a fixed point and its location  is identical to its location in the initial ordering. To summarize, in this example we only need to remember 3 elements (the bold numbers in \eq{eq:naive_kv_eff} and \eq{eq:naive_kv_eff_2}) instead of the original $6+5+4+3=18$.

\T{Individual \mbox{KV} entries.} To leverage  this solution with reduced memory requirements, we stop organizing the key-value store using arrays. Instead of keeping an entry $\mbox{KV}[b][h]$, we keep an entry $\mbox{KV}[(b,h)]$ where the pair $(b,h)$ is the key. This can be efficiently implemented by simply 
concatenating $b$ and $h$ to form a single key. For example, in \eq{eq:naive_kv_eff_2}, instead of using $\mbox{KV}[0][1]{=}4$ with an array, we use $\mbox{KV}[(0,1)]{=}4$. Thus,
\begin{equation*}
H_{\mathcal{W}_b}(k) =  
\begin{cases}
\mbox{KV}[(b, h_b(k))] &\quad\text{if the entry exists} \\
h_b(k) &\quad\text{otherwise}
\end{cases}
\end{equation*}

To efficiently determine the desired order within $\mathcal{W}_b$ for a newly removed bucket $b$ and the exact elements that we need to store, we maintain two additional arrays: (1)~$W$, which always contains the \emph{current} set of working buckets in their desired order, and (2)~$L$, which stores for each bucket its \emph{most recent} location in $W$. Both arrays are initialized identically: $W[b]=L[b]=b \quad\forall b \in \set{0,1,\ldots,a-1}.$ For instance, after bucket $1$ is removed (\ie last array in \eq{eq:naive_kv_eff}), $W$ and $L$ obtain the following form:
$$
\overset{^{\footnotesize\mbox{$b$:}}}{\mbox{$W[b]$:}} \quad \overset{0}{\fbox{0}} \sep \overset{1}{\fbox{4}} \sep \overset{2}{\fbox{2}} \sep \overset{3}{\fbox{3}} \sep \overset{4}{\fbox{4}} \sep \overset{5}{\fbox{5}} \sep \overset{6}{\fbox{6}}
\quad
\overset{^{\footnotesize\mbox{$b$:}}}{\mbox{$L[b]$:}} \quad \overset{0}{\fbox{0}} \sep \overset{1}{\fbox{1}} \sep \overset{2}{\fbox{2}} \sep \overset{3}{\fbox{3}} \sep \overset{4}{\fbox{1}} \sep \overset{5}{\fbox{5}} \sep \overset{6}{\fbox{6}}.
$$
That is, bucket $4$ replaced bucket $1$ in $W$ and the most recent location of bucket $4$ updated to index $1$. 
Note that the removals of buckets $5$ and $6$ did not require any updates in both $W$ and $L$. With this example at hand, we now detail the update rules for $W$ and $L$ upon bucket removals and additions.

\T{Removal.} Assume a newly removed bucket $b$ and let $N=\vert \mathcal{W}_b \vert$. Then in $W$, $b$ is replaced by the last positioned working bucket (\ie $W[N]$), and its most recent location (\ie $L[b]$) is correspondingly updated in $L$. This yields 
$$W[L[b]] \gets W[N], \,\, L[W[N]] \gets L[b].$$ 

Now, we use the updated array $W$ to determine which entries to store in \mbox{KV}: for all $h \in \set{0,1,\ldots,|\mathcal{W}_b|-1}$ such that $W[h] \neq h$, we store $\mbox{KV}[(b,h)] \gets W[h]$.
 
\T{Addition.} Upon bucket addition, we need to \emph{restore} the state prior to the last removal. To do so, we delete the corresponding entries in \mbox{KV} by the same rule we used to remember them. Then, we restore $W$ and $L$ to their previous state using: 
$$L[W[N]] \gets N, \,\, W[L[b]] \gets b.$$
For example, given the state in \eq{eq:naive_kv_eff}, if we now add back bucket $1$ then we simply restore $W$ and $L$ to their initial state, since using the rules yields $L[4] \gets 4$ and  $W[1] \gets 1$.

\T{Complexity.} In the worst case, each consecutive removed bucket may require one additional entry in addition to the entries required by the previously removed bucket. Accordingly, this method for resolving $H_{\mathcal{W}_b}(k)$ results in a memory footprint of 
$1+2+\ldots+|\mathcal{R}|=0.5|\mathcal{R}|(|\mathcal{R}|+1)$, together with three arrays of size (at most) $|\mathcal{A}|$ to represent $A$, $W$ and $L$.
The total memory footprint is therefore $\Theta(|\mathcal{A}|+|\mathcal{R}|^2)$. Updating $\mbox{KV}$ upon a bucket addition or removal incurs a complexity of $O(|\mathcal{W}|)$.

\subsection{Minimal-memory implementation}

While the previous implementation may be sufficient for systems with a small $|\mathcal{R}|$ value, we present our final implementation of \name that results in a remarkably low-memory footprint, negligible response time to changes and high key lookup rate. Specifically, we show how to efficiently calculate $\mbox{KV}[(b,h)]$ for all $(b,h)$ pairs, using a single array that replaces the key-value store functionality.

\T{Successors.} 
To do so, for each removed bucket $b$, we are only storing its \emph{successor}, \ie the bucket that replaced it in $W$. That is, we define an array $K$, such that its entry for each removed bucket $b$ is $K[b]=\mbox{KV}[(b,L[b])]$. We initiate $K[b] \gets b \,\, \forall b \in \set{0,1,\dots,a-1}$, as initially a working bucket $b$ appears at $W[b]$ (\ie replaces itself). For instance, in the example of \eq{eq:naive_kv_eff}, we just remember that bucket 4 replaced bucket 1 (\ie $K[1]=4$), and later in the example of \eq{eq:naive_kv_eff_2}, that bucket 3 replaced bucket 0 (\ie $K[0]=3$). This yields

\begin{align}\label{eq:first_k_exp}
\overset{^{\mbox{$b$:}}}{K[b]} &=  \quad \overset{0}{\fbox{\textbf{3}}} \sep \overset{1}{\fbox{\textbf{4}}} \sep \overset{2}{\fbox{2}} \sep \overset{3}{\fbox{3}} \sep \overset{4}{\fbox{4}} \sep \overset{5}{\fbox{5}} \sep \overset{6}{\fbox{6}}, \cr
A[b] &=  \quad \fbox{3} \sep \fbox{4} \sep \fbox{0} \sep \fbox{0} \sep \fbox{0} \sep \fbox{5} \sep \fbox{6}. 
\end{align}

We next show that we can exploit this information to reconstruct the individual \mbox{KV} entries used by the reduced-memory footprint implementation. Our key observation is that when trying to resolve $\mbox{KV}[(b,h)]$, we are actually searching for $W[h]$ just after $b$'s removal. Therefore, we can trace through the history of $W[h]$, until we reach $\mathcal{W}_b[h]$. We start from $h$, which is the initial value of $W[h]$. When bucket $h$ was removed, $W[h]$ was updated to its successor, \ie $K[h]$. Similarly, when $K[h]$ was removed, it was updated to its successor as well, \ie $K[K[h]]$, and so on. Accordingly, we iteratively set $h \gets K[h]$, until we reach the first working bucket at $W[h]$ just after $b$'s removal. We determine the stopping condition by looking at the sizes of $\mathcal{W}_b$ and $\mathcal{W}_h$: when $A[b] \geq A[h]$ we know that $K[h]$ was working when $b$ was removed, and can terminate. 

\U{Example.} Consider the example in \eq{eq:first_k_exp}.  

If in this state we further remove bucket $4$, 
we obtain $\mathcal{W}_4 = \{ 2,3 \}$ and 
\begin{align}\label{eq:first_k_exp_2}
\overset{^{\mbox{$b$:}}}{K[b]} &=  \quad \overset{0}{\fbox{\textbf{3}}} \sep \overset{1}{\fbox{\textbf{4}}} \sep \overset{2}{\fbox{2}} \sep \overset{3}{\fbox{3}} \sep \overset{4}{\fbox{\textbf{2}}} \sep \overset{5}{\fbox{5}} \sep \overset{6}{\fbox{6}},
\cr
A[b] &=  \quad \fbox{3} \sep \fbox{4} \sep \fbox{0} \sep \fbox{0} \sep \fbox{2} \sep \fbox{5} \sep \fbox{6}. 
\end{align}
As an example, we show how the arrays $K$ and $A$ can be used to calculate $\mbox{KV}([5,1])$, $\mbox{KV}([1,1])$, and $\mbox{KV}([4,1])$. Recall that $A[b]$ holds $\vert \mathcal{W}_b \vert$ and $\mbox{KV}([b,h])$ holds $H_{\mathcal{W}_b}(k)$, where $h = h_b(k)$. 
Also, if $A[h] < A[b]$, then $h \in \mathcal{W}_b$ ($h$ was not removed before $b$) and $\mbox{KV}([b,h]) = h$.

We use $A$ to check if $\mbox{KV}([b,h]) = h$. In our case, $\mbox{KV}([5,1]) = 1$, but $\mbox{KV}([4,1]) \neq 1$ and $\mbox{KV}([1,1]) \neq 1$, since $A[4] < A[1] < A[5]$. In other words, we know from $A$ that $1 \in \mathcal{W}_5$, but $1 \not\in \mathcal{W}_4$ (and, obviously, $1 \not\in \mathcal{W}_1$). We use $K$ to calculate $\mbox{KV}([1,1])$; since $K[1] = 4$ and $4 \in \mathcal{W}_1$ ($A[4]<A[1]$), we conclude that $\mbox{KV}([1,1])=4$. 
Similarly, we use $K$ twice to calculate $\mbox{KV}([4,1])$; since $K[1] = 4$ and $4 \not\in \mathcal{W}_4$, we examine $K[K[1]] = 2$. Since $2 \in \mathcal{W}_4$ ($A[2]<A[4]$), we conclude that $\mbox{KV}([4,1]) = 2$.

All three of the above calculations may be needed to find a working bucket  by \textsc{GetBucket($k$)}. For example, consider a key $k$ for which $H_{\mathcal{A}}(k) =5$ and $h_5(k)= h_1(k)= h_4(k) = 1$. First we examine bucket $5$ and since $5 \not\in \mathcal{W}$ ($A[5]>0$) we rehash to $H_{\mathcal{W}_5}(k) = \mbox{KV}([5,h_5(k)]) = \mbox{KV}([5,1]) =1$. Since $1 \not\in \mathcal{W}$, we rehash to $H_{\mathcal{W}_1}(k) = \mbox{KV}([1,h_1(k)]) = \mbox{KV}([1,1]) =4$. Since $4 \not\in \mathcal{W}$, we rehash yet again to $H_{\mathcal{W}_4}(k) = \mbox{KV}([4,h_4(k)]) = \mbox{KV}([4,1]) =2$. Finally, since $2 \in \mathcal{W}$ it can be returned as the bucket for key $k$.

\begin{algorithm}[!t]
	\footnotesize
    \begin{algorithmic}[1]
    
     	\Function{InitAnchor}{$a, w$}
          \State $A[b] \gets 0 \text{ for }b=0,1,\ldots, a{-}1$ \Comment{$|\mathcal{W}_b| \gets 0$ for $b \in \mathcal{A}$}
          \State $R \gets \emptyset$ \Comment{Empty stack} 
          \State $N \gets w$ \Comment{Number of initially working buckets}
   		  \State $K[b] \gets L[b] \gets W[b] \gets b \text{ for }b=0,1,\ldots, a-1$
          \For{$b = a{-}1$ \textbf{downto} $w$} \Comment{Remove initially unused buckets}
              \State $R.push(b)$
              \State $A[b] \gets b$
          \EndFor
        \EndFunction
        
    \\\hrulefill    
    
        \Function{GetBucket}{$k$}
        	\State $b \gets \mbox{hash}(k) \bmod a$
        	\While {$A[b] {>} 0$} \Comment{$b$ is removed}
        		\State $h \gets h_b(k)$
        		\Comment{$h \gets \mbox{hash}(b,k) \bmod A[b]$}
                \While {$A[h] \ge A[b]$} 
                \Comment{$\mathcal{W}_b[h] \not = h$, $b$ removed prior to $h$}
					\State $h \gets K[h]$ 
					\Comment{search for $W_b[h]$}
                \EndWhile                
               	\State $b \gets h$
               	\Comment{$b \gets H_{\mathcal{W}_b}(k)$}
        	\EndWhile
        	\State\Return $b$
        \EndFunction
        
       \\\hrulefill    

        \Function{AddBucket}{ }
			\State $b \gets R.pop()$
            \State $A[b] \gets 0$ \Comment{$\mathcal{W} \gets \mathcal{W} \cup \{b\}$, delete $\mathcal{W}_b$}
        	\State $L[W[N]] \gets N$
            \State $W[L[b]] \gets K[b] \gets b$
            \State $N \gets N+1$
            \State\Return $b$
        \EndFunction
        
    \\\hrulefill
    
        \Function{RemoveBucket}{$b$}      
            \State $R.push(b)$
        	\State $N \gets N-1$
            \State $A[b]  \gets  N$ \Comment{$\mathcal{W}_b \gets \mathcal{W} \backslash b$, $A[b] \gets |\mathcal{W}_b|$}
			\State $W[L[b]] \gets K[b] \gets W[N]$
            \State $L[W[N]] \gets L[b]$
        \EndFunction
       
    \end{algorithmic}   
	\normalsize
    \caption{ --- AnchorHash Implementation}
    \label{alg:AnchorHash_impl}
\end{algorithm}

\T{Complexity.} \al{alg:AnchorHash_impl} provides the pseudo-code for \name's final array-based implementation. The memory footprint for this solution is $\Theta (|\mathcal{A}|)$ \emph{independently of the system state}--- \eg, independently of the number of removed buckets or of their identity. That is, we keep four arrays $A, L, W, K$ of size $|\mathcal{A}|$ and the stack $\mathcal{R}$. The update time upon a bucket removal or addition accounts for $O(1)$ operations and is negligible for any $\mathcal{A}$ and $\mathcal{R}$. Note that $\Theta (|\mathcal{A}|)$ is required to save resource details (\eg server IP addresses). 

While we already established bounds on the number of \emph{hash operations}, we now provide an upper bound on the average number of \emph{memory accesses} required by a key lookup when using our final minimal-memory implementation.

\begin{theorem}[Memory accesses]\label{comp_complexity_en} Assume \emph{random} removals. Let $\vert \mathcal{W} \vert=w$ and $\vert \mathcal{A} \vert=a$. 
Denote by $\xi$ the total number of memory accesses performed by \textsc{GetBucket($k$)} for a randomly chosen key $k$ when using the minimal-memory implementation. Then,
the average of $\xi$ is O\big($\para{1+\ln\para{\frac{a}{w}}}^2$\big). 
\end{theorem}

\begin{proof}
Denote by $\mathcal{R}$ the (random) sequence of bucket removals. We will prove the result recursively on the size of $\mathcal{R}$. 
Suppose $|\mathcal{R}|=r$ buckets were randomly removed, and now we randomly remove an additional bucket. 

By the minimal disruption property, only the keys that were mapped to this newly removed bucket are remapped. Likewise, by the balance property, this occurs with probability $\frac{1}{a-r}$ for a randomly chosen key. The keys that are remapped require an additional access to the array $A$ and possibly the \emph{resolution} of the bucket's identity using the array \mbox{K}, where the latter depends on the sequence of removals and the last index the key hits. 

The length of the required resolution when hitting index $i$ is upper-bounded by the number of times the bucket associated with index $i$ was removed. Denote this quantity by $\Delta_i^r$.Since the remapped keys have an equal probability of hitting any index in $\set{0,1,\ldots,a-r-2}$, we obtain
\begin{align}\label{eq:xi_ub_exp}
\mathbbm{E}[\xi &\mid |\mathcal{R}|=r+1]\leq\mathbbm{E}[\xi \mid |\mathcal{R}|=r]\cr
&+ \cdot \frac{1}{a-r}\Big(1+\frac{1}{a-r-1}\sum_{i=0}^{a-r-2}\mathbbm{E}[\Delta_i^{r+1}]\Big),   
\end{align}
Now, we observe that since index $0$ is always associated with a working bucket, $\Delta_0^{r+1}$ is stochastically larger than $\Delta_i^{r+1}$ for all $i$. Also,
\begin{equation}\label{eq:delta_ub}
\begin{split}
\mathbbm{E}[\Delta_0^{r+1}]&=\frac{1}{a}+\ldots +\frac{1}{a-r-1}\cr
&=\sum_{k=1}^{r+1}\frac{1}{a-k}\leq \ln{\Big(\frac{a}{w}\Big)}.  
\end{split}
\end{equation}
Thus, using \eqref{eq:delta_ub} in \eqref{eq:xi_ub_exp} yields the following recurrence,
\begin{equation}\label{eq:rec_r}
\begin{split}
&\mathbbm{E}[\xi \mid |\mathcal{R}|=r+1] \leq\cr&\mathbbm{E}[\xi \mid |\mathcal{R}|=r]
+\frac{1}{a-r}\Big(1+\ln{\Big(\frac{a}{w}\Big)}\Big), 
\end{split}
\end{equation}
with the initial condition given by 
\begin{equation}\label{eq:rec_initial_cond}
\mathbbm{E}[\xi \mid |\mathcal{R}|=0]=1.
\end{equation}
Now, by solving the recurrence given by \eqref{eq:rec_r} and \eqref{eq:rec_initial_cond} we obtain
\begin{align}
\mathbbm{E}[\xi \mid |\mathcal{R}|=a-w]&\leq 1+\ln{\Big(\frac{a}{w}\Big)}+\ln^2{\Big(\frac{a}{w}\Big)}\cr
&\leq   \Big(1+\ln{\Big(\frac{a}{w}\Big)}\Big)^2,   
\end{align}
which concludes the proof.
\end{proof}


Finally, Table \ref{tab:anchor_evo} summarizes the differences between the naive, reduced-memory, and the final minimal-memory implementations.

\begin{table}[t]
\centering
\begin{tabular}{@{}lC{2.0cm}C{2.2cm}C{1.2cm}@{}C{1.2cm}@{}} 
\toprule
        & \textbf{Hash operations}  & \textbf{Memory accesses}         & \textbf{Memory}             & \textbf{Update}      \\  
\midrule
Naive      &    $O(1+\log{(\frac{a}{w})})$    & $O(1+\log{(\frac{a}{w})})$    &  $\Theta(a+ar)$   &   $\Theta(w)$  \\ 
Reduced    &    $O(1+\log{(\frac{a}{w})})$    & $O(1+\log{(\frac{a}{w})})$    &  $O(a+r^2)$  &   $O(w)$  \\ 
Minimal    &    $O(1+\log{(\frac{a}{w})})$    & $O((1+\log{(\frac{a}{w})})^2)$    &  $\Theta(a)$      &   $O(1)$    \\ 
\bottomrule 
\end{tabular}
\caption{\name implementation evolution: naive, then  reduced-memory, then minimal-memory implementations. Successive implementations reduce the memory footprint to improve scalability, but the final implementation also slightly increases the upper bound on the expected number of memory accesses.}
\label{tab:anchor_evo}
\end{table}


\begin{figure}[!t]
\centering

\begin{subfigure}[c]{\linewidth} \centering
\MyIncludeGraphics[clip, trim=0.1cm 1.1cm 0.1cm 0.1cm, width=0.99\textwidth]{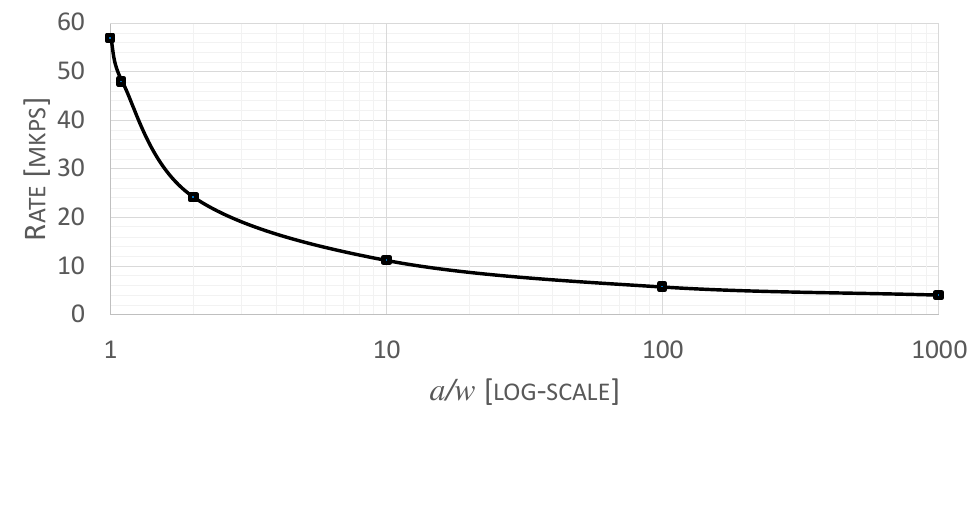}
\caption{\footnotesize \name key lookup rate with 1,000 working buckets with respect to different $a/w$ ratios. 
}
\label{fig:anchor_speed_1000}
\end{subfigure}

\begin{subfigure}[c]{\linewidth} \centering
\MyIncludeGraphics[clip, trim=0.1cm 0.1cm 0.1cm 0.1cm,width=0.99\linewidth]{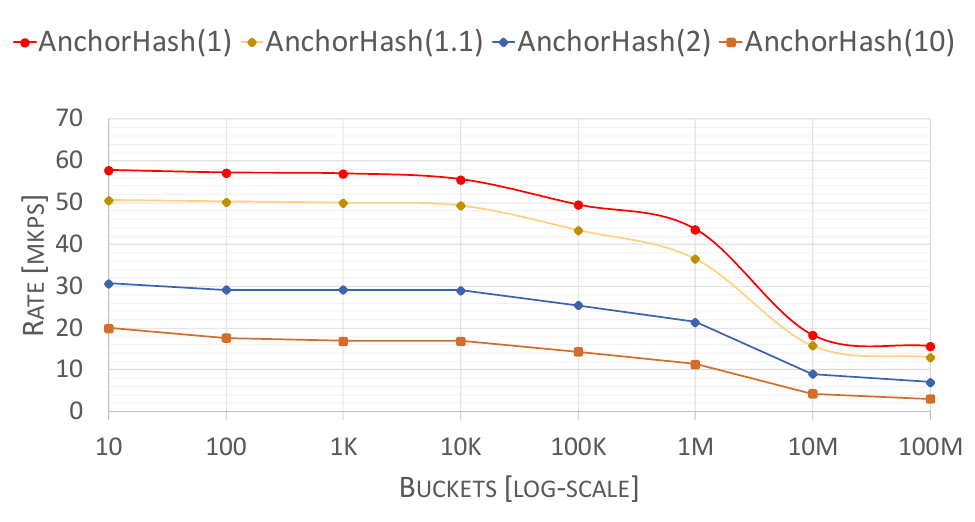}
\caption{\footnotesize \name key lookup rate with respect to the number of working buckets for different fixed $\frac{a}{w}$ ratios. \emph{AnchorHash(x)} stands for an \name instance with $\frac{a}{w}=x$. 
}
\label{fig:anchor_speed_random}
\end{subfigure}

\caption{\name key lookup rate in millions of keys per second (Mkps). 
}
\label{fig:anchor_speed}
\end{figure}


\begin{figure*}[h]
\centering
\MyIncludeGraphics[clip, trim=0.1cm 0.1cm 0.1cm 0.1cm, width=\linewidth]{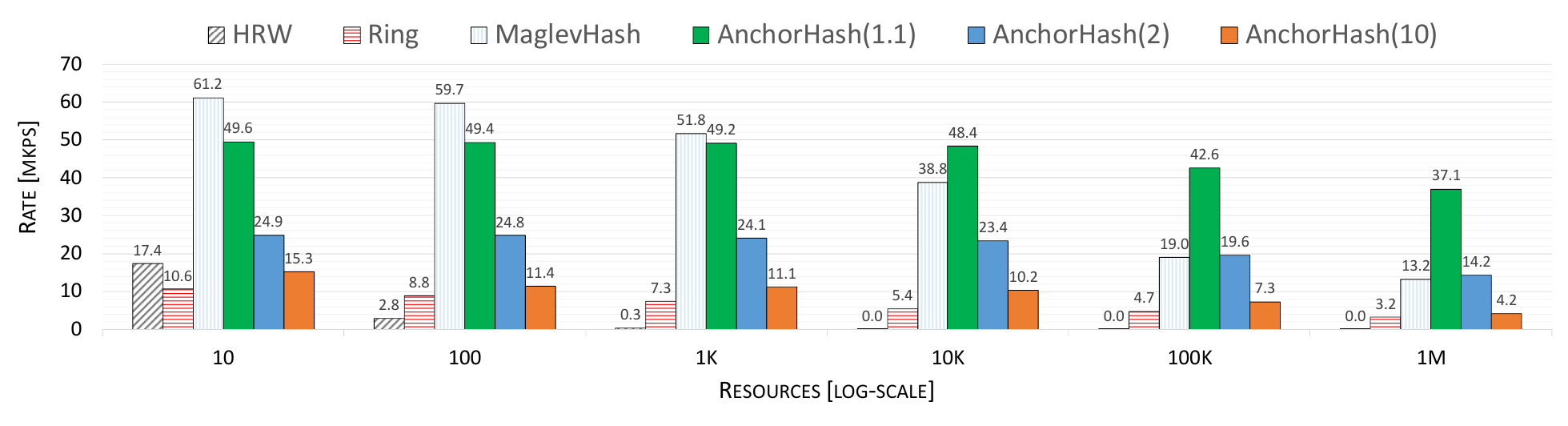}
\caption{\footnotesize Comparing the key lookup rates between HRW, Ring, MaglevHash and \name for different resource counts. Due to the significantly smaller memory footprint, \name maintains an extremely high rate even for $10^5$ resources.}
\label{fig:speed_comp}
\end{figure*}

\begin{figure}[h]
\centering
\MyIncludeGraphics[clip, trim=0.1cm 0.1cm 0.1cm 0.1cm,width=\linewidth]{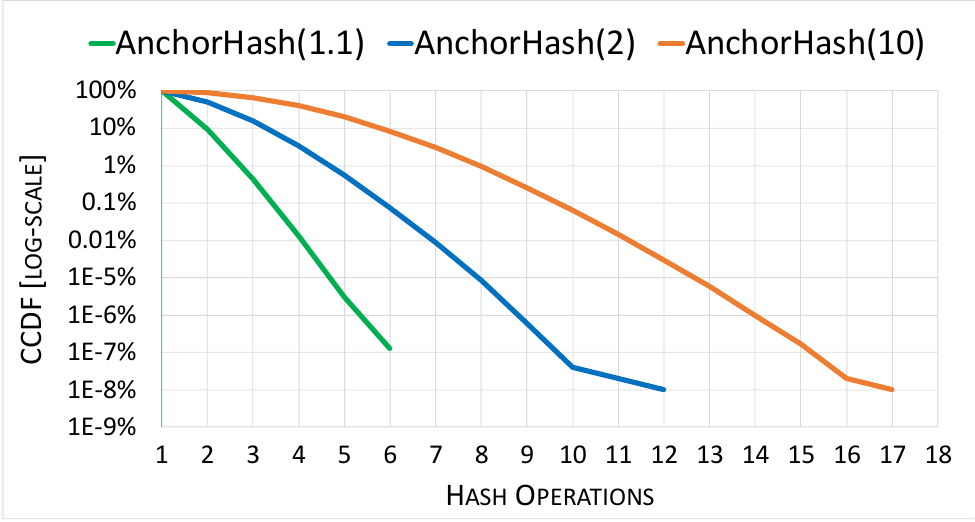}
\caption{\footnotesize CCDF for the number of hash operations performed by \name per key lookup. For example, for AnchorHash$(2)$, $99.9\%$ of keys would require $6$ or less hash operations.}
\label{fig:anchor_key_lookup_hashes}
\end{figure}

\begin{figure*}[h]
\centering
\MyIncludeGraphics[clip, trim=0.1cm 0.1cm 0.1cm 0.1cm,width=\linewidth]{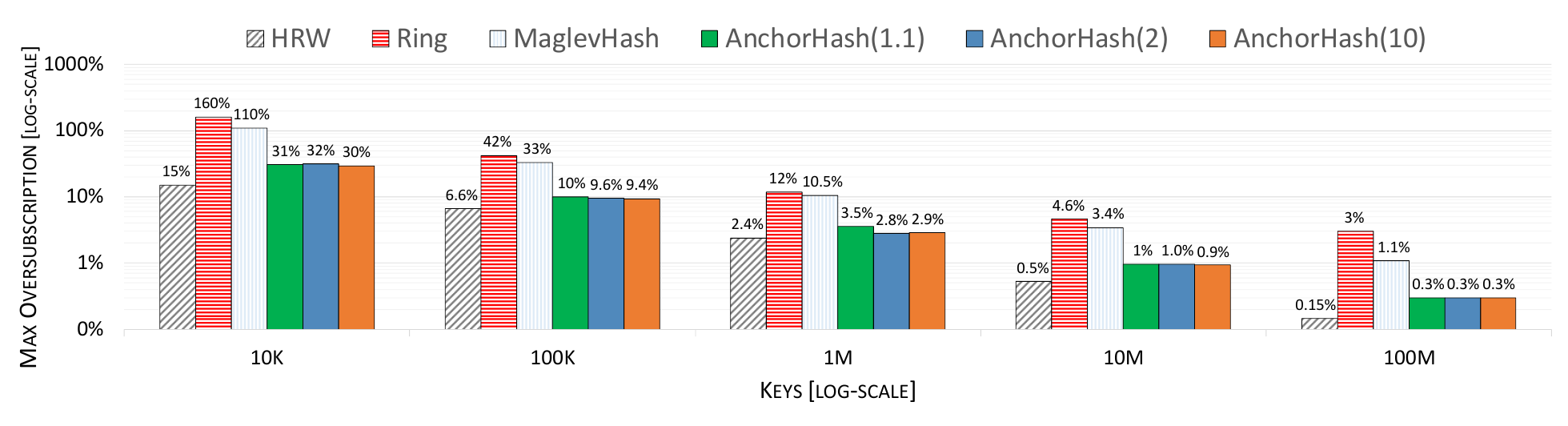}
\caption{
\footnotesize Comparing worst-case oversubscription. Lower is better (better balance). All instances have 1,000 resources. For \name we have an Anchor of 1,100, 2000 and 10000 buckets with 100, 1000 and 9000 random removals accordingly (\ie AnchorHash(1.1), AnchorHash(2) and AnchorHash(10)). 
}
\label{fig:balance_test}
\end{figure*}

\begin{figure*}[h]
\centering
\MyIncludeGraphics[clip, trim=0.1cm 0.1cm 0.1cm 0.1cm, width=\linewidth]{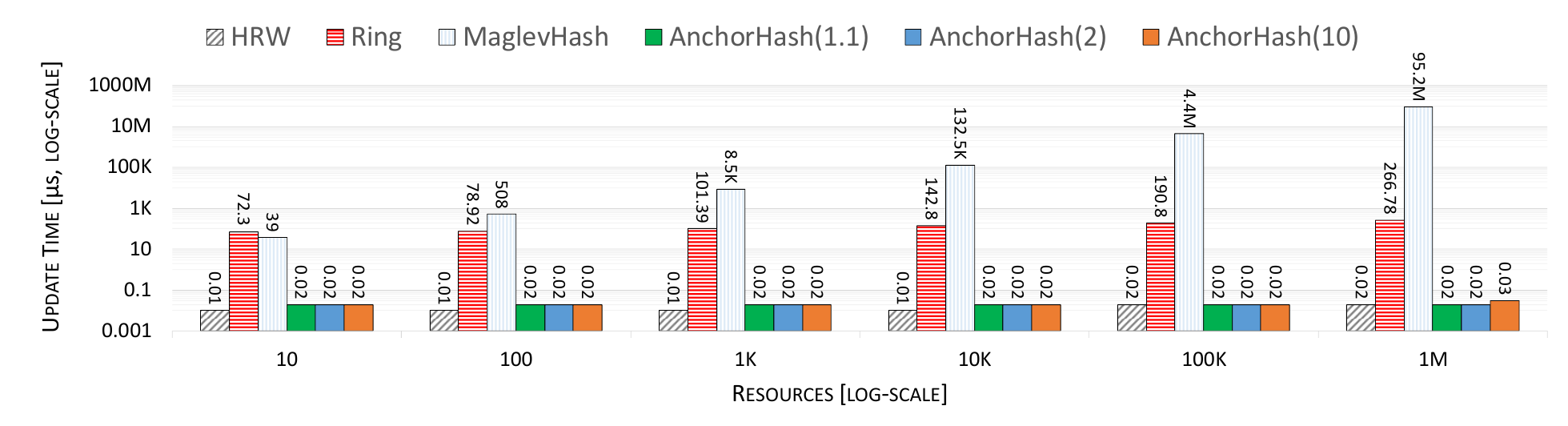}
\caption{%
\footnotesize Comparing update time for resource removals/additions.  HRW and \name require only a few tens of nano-seconds independently of the size of the system. For Ring and especially MaglevHash, the update time increases with the size of the system. For example, for $10^5$ resources, MaglevHash requires more than 4 seconds to repopulate its array.}
\label{fig:rt_test}
\end{figure*}


\section{Evaluation} \label{sec:eval}

\T{Algorithms.} In this section we test and compare \name to  HRW, Ring, and  
MaglevHash, according to the evaluation metrics of Table \ref{tab:main_comp}: consistency (\ie minimal disruption and balance), key lookup rate, memory footprint, and update time upon additions and removals.

\T{Testbed.} All our experiments were conducted on a single core of a commodity machine with an Intel i7-7000 CPU at 3.6 GHz, 16~GB of RAM and an Ubuntu 16.04 LTS operating system. All algorithm implementations are in C++ and are optimized for run-time purposes. 
In our evaluation, each bucket has a 32-bit identifier (\ie up to $2^{32}$ buckets are supported),  and we use 64-bit randomly-generated keys. For all algorithms we use the crc32 \cite{singhal2008inside} hash function with two 64-bit inputs (key and seed) for uniform hashing.

\T{Memory footprint.} Before turning to empirical evaluation, we first discuss the memory footprint of the four approaches, as it has a significant impact on all other qualities such as key lookup rate and update time.

The memory footprint of Ring and MaglevHash depends on the theoretical hash-space balance guarantee these algorithms provide. For example, in MaglevHash, reaching a maximum of 1\% hash space imbalance requires \emph{at least} $\frac{1}{0.01}=100$ copies for each resource. Throughout our evaluation, for MaglevHash and Ring we use $100$ copies for each resource \cite{maglev}. On the other hand, HRW and \name provide perfect hash-space balance and do not require copies to do so. 

In our implementation, \name requires only 16 Bytes of memory per resource. This means that even for $10^6$ resources, \name uses 16 MB of space, whereas MaglevHash requires at least 400 MB to achieve a reasonable balance for the same scenario.

\T{Lookup rate.} We test 
\name's key lookup rate for different Anchor sizes (up to $10^8$) and different $\frac{a}{w}$ ratios (up to $10^3$). For example,  $w=1,000$ and $\frac{a}{w}=100$ means that only $w=1,000$ resources are still active out of $a=100,000$ (\ie a scenario with 99,000 random removals).

The results  are depicted in \fig{fig:anchor_speed}. \fig{fig:anchor_speed_1000} shows the key lookup rate achieved by \name with 1,000 working buckets with respect to different $\frac{a}{w}$ ratios. \fig{fig:anchor_speed_random} depicts \name rate with respect to the number of working buckets for different fixed $\frac{a}{w}$ ratios. Note that, even for a fixed $\frac{a}{w}$ ratio, the rate slightly decreases as the number of buckets increases. This is because of the increased percentage of L3 cache misses as follows from the increased memory footprint.
Remarkably, even for a million buckets, \name achieves a rate of tens of millions of keys per second for reasonable and even extreme operating points (\eg half of the buckets have been randomly removed).



Next, Fig.~\ref{fig:speed_comp} compares the key lookup rates achieved by the four approaches for different number of resources. 
For \name, we depict three scenarios with 10\%, 50\% and 90\% random removals, corresponding to AnchorHash(1.1), AnchorHash(2) and AnchorHash(10). 
AnchorHash(1.1) reaches a high key lookup rate that is similar to MaglevHash.
As the resource count increases, MaglevHash suffers from a more significant rate degradation due to increased L3 cache misses that stem from its much larger memory footprint. On the other hand, as expected, the rate of \name decreases for higher percentages of random removals, due to the larger number of hash computations.

Additionally, we tested the lookup rate of the four approaches using a backbone router CAIDA trace \cite{CAIDACH16}. 
The results follow similar trends. Interestingly, all approaches run faster since the often reoccurring flow packets increase the cache hit rate. 

We also measured the number of hash operations for a key lookup of \name with 1,000 resources with an Anchor of 1,100, 2000 and 10000 buckets with 100, 1000 and 9000 random removals respectively (\ie AnchorHash(1.1), AnchorHash(2) and AnchorHash(10)). The number of simulated keys is $10^8$. The results are depicted in \fig{fig:anchor_key_lookup_hashes}. For all versions of AnchorHash, it is evident that the number of hash operations is exponentially decreasing. Moreover, while the worst case in terms of hash operations is 101, 1001 and 9001 hash operations for the three versions of AnchorHash, out of $10^8$ keys no key required more than 6, 12 and 17 hash operations respectively. In AnchorHash(1.1), more than 90\% of keys terminate after a single hash operation and less than 0.5\% require more than 2. Even for AnchorHash(10), 99\% terminate with less than 7 operations. 

\T{Balance.} Essentially, there are three sources of imbalance, all reflected in an algorithm's load-balancing abilities: 
(1)~hash space imbalance; (2)~quality of the hash function; and (3)~arriving keys. While the last two are implementation- and workload-dependent, the first  is algorithm-dependent. Thus, in terms of balance, assuming uniform hashing, HRW and \name have an inherent advantage over MaglevHash and Ring. 
To demonstrate this, we tested the four approaches using the same hash function and a random stream of keys. 

By standard practice \cite{maglev} we measure the worst-case resource oversubscription in \%. For instance, an oversubscription of 10\% means that the most loaded resource has 10\% more load than the average. All instances run with 1,000 resources. For \name we have an Anchor of 1,100, 2000 and 10000 buckets (corresponding to 100, 1000 and 9000 random removals respectively, \ie AnchorHash(1.1), AnchorHash(2) and AnchorHash(10)). Ring and MaglevHash both run with 100 copies per resource. The results are depicted in \fig{fig:balance_test}. 
As expected, the oversubscription improves for all algorithms as the number of keys increases. The oversubscription of MaglevHash and Ring are theoretically lower-bounded. Specifically, for MaglevHash it is at least 1.01 with 100 copies per resource and for Ring it is the inherent imbalance created by different size intervals in the ring (with high probability). Since HRW and AnchorHash (for all its versions) theoretically provide perfect balance, by the Law of Large Numbers, the oversubscription approaches zero as the number of keys increases. All three versions of AnchorHash have almost the same oversubscription indicating that the size of the anchor has no effect on the resulting balance (as expected). Note that HRW converges to 0 slightly faster than AnchorHash due to the large number of hash operations performed for each key, leading to better randomization.

\T{Update time.} We next test for the time it takes to update the data structure of each of the algorithms with a newly added or removed resource. The results are averaged over 100 trials, and depicted in \fig{fig:rt_test}. 
Both HRW and \name respond in nanosecond scale nearly independently of the size of the system. On the other hand, Ring and MaglevHash respond slower as the system size increases. For example, with $10^5$ resources, MaglevHash requires more than 4 seconds to respond.



\T{Minimal disruption.} We also test the minimal-disruption property for all approaches. Following theory, HRW, Ring and \name achieve the minimal-disruption property in practice as well. Unfortunately, MaglevHash fails to achieve minimal disruption and therefore is not fully consistent. For example, in a scenario with 900 resources and 100 consecutive resource additions, we find that \emph{at each resource addition}, MaglevHash wrongfully reassigns a near-constant fraction of $\approx$ 0.6\% of the hash space, \ie $\approx$ 0.6\% of the keys are needlessly remapped at each of the 100 resource additions. While such \emph{flips} may be acceptable when used together with key tracking (\eg connection tracking in datacenter load-balancing), they may not be acceptable in other systems such as cache servers. 


\section{Conclusion}

In this paper we introduced \name, a new consistent hashing technique. We provided implementation details and theoretical guarantees for \name. We then conducted evaluations comparing \name to existing algorithms. Evaluation results indicate that \name is a scalable and a fully-consistent hashing technique. It is capable of handling millions of resources while maintaining high key lookup rate, low memory footprint, and small update times upon resource additions and removals. Finally, the code for AnchorHash appears in \cite{anchorcode}.

\section{Future Work}
Unlike other approaches, \name leverages state information to achieve its properties. Thus, in a distributed environment (\ie where multiple dispatchers run \name in parallel), \name requires an agreement on the removal order (\ie the content of $\mathcal{R}$) to ensure full consistency. While this overhead is small in terms of communication overhead (happens only once upon removal), it is of interest to study whether \name can be extended to maintain full consistency in a setting in which the dispatchers do not necessarily agree on the order of removals.

\section*{Acknowledgments}

This work was partly supported by the Hasso Plattner Institute Research School, the Israel Science Foundation (grant No. 1119/19),  the Technion Hiroshi Fujiwara Cyber Security Research Center, and the Israel Cyber Bureau.



\bibliographystyle{IEEEtran}
\bibliography{bib}


\vspace{3cm}
\enlargethispage{-3cm}
\begin{IEEEbiography}[{\includegraphics[width=1in,height=1.25in,clip,keepaspectratio]{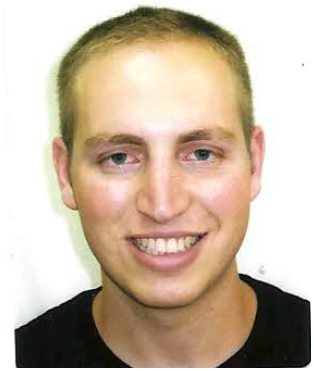}}]%
	{Gal Mendelson} received his BSc, MSc (summa cum laude) and Ph.D. degrees from the Viterbi department of Electrical Engineering, Technion -– Israel Institute of Technology, in 2009, 2015 and 2020, respectively. He was the recipient of the Hasso Plattner Institute Ph.D. fellowship award and the INFORMS Applied Probability Society best student paper award. He is mainly interested in stochastic analysis, algorithms and communication networks. 
\end{IEEEbiography}
\begin{IEEEbiography}[{\includegraphics[width=1in,height=1.25in,clip,keepaspectratio]{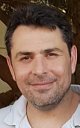}}]%
	{Shay Vargaftik} received his B.Sc. and Ph.D. degrees from the Viterbi department of Electrical Engineering, Technion -– Israel Institute of Technology, in 2012 and 2019, respectively. He was the recipient of the Hasso Plattner Institute and the IBM Ph.D. fellowship awards. He is currently a postdoctoral researcher in the VMware Research Group (VRG). He is mainly interested in the theory and practice of networking and machine learning with an emphasis on scalability and efficient resource usage.
\end{IEEEbiography}
\begin{IEEEbiography}[{\includegraphics[width=1in,height=1.25in,clip,keepaspectratio]{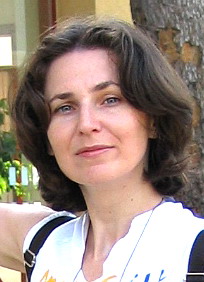}}]%
	{Katherine Barabash} received her B.Sc. in Applied Mathematics and
M.Sc. in Computer Science degrees from the Technion -- Israel Institute of Technology, in 1994 and 2010 respectively. Kathy is a researcher in IBM Research since 1997 and has contributed to system research in areas of memory management, storage, software defined networking, as well as other data center and Cloud technologies. Kathy's current research is devoted to Hybrid Cloud and 5G networking.  
\end{IEEEbiography}
\enlargethispage{-3cm}
\begin{IEEEbiography}[{\includegraphics[width=1in,height=1.25in,clip,keepaspectratio]{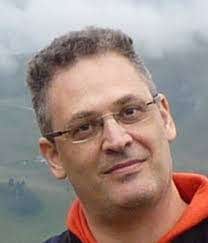}}]%
	{Dean H. Lorenz} received his B.Sc. (summa com laude) in Computer Engineering and Ph.D. in Electrical Engineering, from the Technion, Haifa, Israel. 
	He is Researcher at IBM Research -- Haifa, where he is a technical leader in the Cloud Architecture Networking group, in the Hybrid Cloud department. 
	Dr. Lorenz has more than 20 years of experience in research, hands-on development, and innovation in Networking, Virtualization, Storage, and Mobile Technologies; and has held technical positions at leading companies in these industries, including IBM Research, Akamai, Adobe Omniture, and Qualcomm. His current research is Cloud technologies, with focus on Cloud networking, AIOps, elasticity, and operation efficiency.
\end{IEEEbiography}
\begin{IEEEbiography}[{\includegraphics[width=1in,height=1.25in,clip,keepaspectratio]{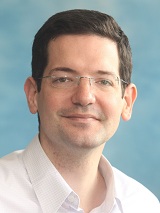}}]%
	{Isaac Keslassy} (M{'}02, SM{'}11) received his M.S. and Ph.D. degrees in Electrical Engineering from Stanford University, Stanford, CA, in 2000 and 2004, respectively. He is currently a full professor in the Viterbi department of Electrical Engineering at the Technion, Israel. His recent research interests include the design and analysis of data-center networks and high-performance routers.  He was the recipient of an ACM SIGCOMM test-of-time award, of an ERC Starting Grant, and of the Allon, Mani, Yanai, and Taub awards. He was associate editor for the IEEE/ACM Transactions on Networking.
\end{IEEEbiography}
\begin{IEEEbiography}[{\includegraphics[width=1in,height=1.25in,clip,keepaspectratio]{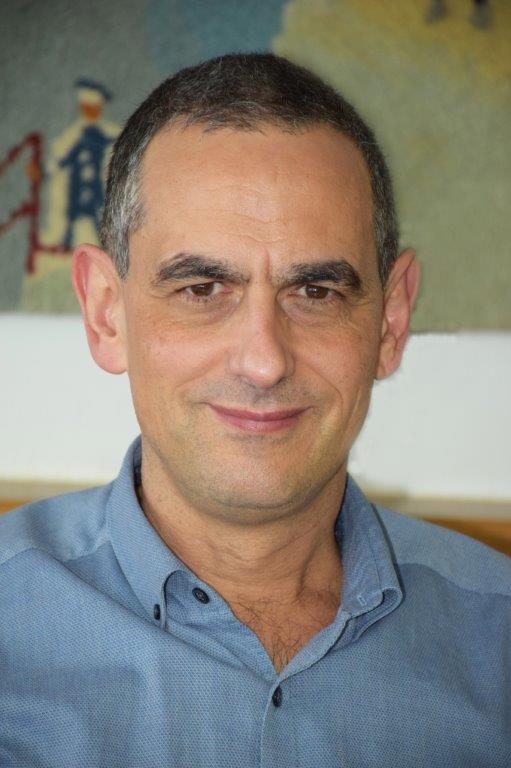}}]%
	{Ariel Orda} (S{'}84, M{'}92, SM{'}97, F{'}06) received the BSc (summa cum laude), MSc, and DSc degrees in electrical engineering from the Technion, Haifa, Israel, in 1983, 1985, and 1991, respectively. During 1.1.2014-31.12.2017, he was  the dean of the Viterbi Department of Electrical Engineering, Technion. Since 1994, he has been with the Department of Electrical Engineering, Technion, where he is the Herman and Gertrude Gross professor of communications. His research interests include network routing, the application of game theory to computer networking, survivability, QoS provisioning, wireless networks, and network pricing. He served as program co-chair of IEEE INFOCOM 2002, WiOpt 2010 and Netgcoop 2020, and general chair of Netgcoop 2012. He was an editor of the IEEE/ACM Transactions on Networking and Computer Networks. He received several awards for research, teaching, and service.
\end{IEEEbiography}
\end{document}